\documentclass[times]{elsarticle}

\usepackage{jcomp}
\usepackage{framed,multirow}

\usepackage{amssymb}
\usepackage{latexsym}

\usepackage{url}
\usepackage{xcolor}
\definecolor{newcolor}{rgb}{.8,.349,.1}
\usepackage{booktabs}
\usepackage{bm}
\usepackage{amsmath}
\usepackage{amsthm}
\usepackage{algorithm}
\usepackage{algpseudocode} 
\usepackage{float}
\usepackage{lineno}

 \usepackage{url}
\usepackage{hyperref}
\hypersetup{
     unicode=false,          
    pdftoolbar=true,        
    pdfmenubar=true,        
    pdffitwindow=false,     
    pdfstartview={FitH},    
    pdftitle={My title},    
    pdfauthor={Author},     
    pdfsubject={Subject},   
    pdfcreator={Creator},   
    pdfproducer={Producer}, 
    pdfkeywords={keyword1} {key2} {key3}, 
    pdfnewwindow=true,      
    colorlinks=true,       
    linkcolor=blue,          
    citecolor=blue,        
    filecolor=magenta,      
    urlcolor=cyan           
}

\usepackage{todonotes}

\newcommand{\kbt}{k_B T}

\newcommand{\LLL}{\mathcal{L}}
\newcommand{\AAA}{\mathcal{A}}
\newcommand{\BBB}{\mathcal{B}}
\newcommand{\SSS}{\mathcal{S}}
\newcommand{\KKK}{\mathcal{K}}
\newcommand{\III}{\mathcal{I}}
\newcommand{\pp}{\partial}
\newcommand{\nnorm}[1]{\Vert #1 \Vert }
\newcommand{\dd}{\mathop{}\!\mathrm{d}}
\newcommand{\dt}{\mathop{}\!\mathrm{dt}}

\newcommand{\real}{\mathbb{R}}

\graphicspath{{figures/}{images/}}
\newcommand{\xx}{{\bm{x}}}
\newcommand{\XX}{{\bm{X}}}
\newcommand{\vv}{{\bm{v}}}
\newcommand{\VV}{{\bm{V}}}

\newtheorem{theorem}{Theorem}

\journal{Journal of Computational Physics}

\begin{document}

\verso{T. Li \textit{etal}}

\begin{frontmatter}

\title{Iterative Neural Variational Solver of Committor Function in Under-damped Langevin Dynamics}%

\author[1,2,3]{Tiejun Li} 
\ead{tieli@pku.edu.cn}
\author[1]{Zhaoqi Liu}
\ead{lzq8219@stu.pku.edu.cn}
\author[4]{Xiang Zhou}
\ead{xizhou@cityu.edu.hk}

\address[1]{LMAM and School of Mathematical Sciences, Peking University, Beijing 100871, China}

\address[2]{Center for Data Science, Peking University, Beijing 100871, China}

\address[3]{Center for Machine Learning Research, Peking University, Beijing 100871, China}
\address[4]{Department of Mathematics,
City University of Hong Kong,
Hong Kong SAR, China}

\begin{abstract}
The committor function describes the probability that a stochastic dynamical system, while transitioning between two metastable states, will first reach one state rather than the other. It is a central object in transition path theory, providing the accurate reaction rate and the characterization of the transition mechanism through reactive probability flux and the transition state ensemble. While the numerical methods for this function are well studied for over-damped Langevin dynamics where the friction damping coefficient goes to infinity, its computation in the general under-damped case remains an open challenge due to the lack of a variational structure in the underlying backward Kolmogorov equation. This restricts the quantitative investigation of the impact of the finite   damping  coefficient on rare transitions in real systems.
In this work, by decomposing the infinitesimal generator of the underdamped Langevin dynamics into a self-adjoint and a skew-adjoint part, we develop an iterative scheme in which each subproblem has a variational form and can be solved efficiently using neural networks. Fundamentally, this approach can be viewed as a generalized Concus-Golub-Widlund (CGW) splitting method and each subproblem has a variational form, which can be solved efficiently. This iterative scheme also has the flexibility of introducing a preconditioning operator associated with a new overdamped Langevin dynamics, which naturally leads to the design of enhanced adaptive sampling for the subproblems.
We validate the efficiency and accuracy of our method on benchmark systems including a one-dimensional double-well potential, the Müller–Brown potential, and alanine dipeptide. The proposed iterative variational framework provides a practical and extensible numerical tool for probing rare transition events in Langevin dynamics across varying damping coefficients.
\end{abstract}


\end{frontmatter}


 
\section{\label{introduction}Introduction}

One of the most fundamental problems in stochastic systems is understanding the transition between two long-lived metastable states \cite{okuyama1998transition,vanden2006towards,berteotti2009protein,vanden2010transition}. However, many important transitions in chemical reactions, protein folding, and materials design occur on time scales that are currently infeasible for direct microscopic simulation. Transition path theory \cite{vanden2006towards,vanden2010transition} (TPT) provides a framework for analyzing the transition-path ensemble and transition rates. A central object in TPT is the committor function $q(\xx)$ defined as follows. Consider the following stochastic differential equation (SDE):
\begin{equation}\label{general_SDE}
    \dd \XX_t = \bm{b}(\XX_t)\dt + \sigma(\XX_t)\dd \bm{W}_t
\end{equation}
where \(\XX_t\in\Omega\subset\real^{n}\), \(\bm{b}:\real^{n}\to\real^{n}\) and \(\sigma:\real^{n}\to\real^{n\times n}\) are the drift and diffusion functions, respectively, and \(\bm{W}_t\) is an \(n\)-dimensional standard Wiener process. 
 Given two metastable sets $A,B \subset \Omega$, the committor function describes the probability that a trajectory starting at $\xx$ reaches state $B$ before state $A$:
\begin{equation}\label{definition_committor}
    q(\xx) = \mathbb{P}(\tau_B < \tau_A \mid \xx_0 = \xx)
\end{equation}
where $\tau_S = \inf\{t>0:\XX_t \in S\}$ is the hitting time of the set $S\subset\Omega$. The committor determines reaction rates, identifies optimal reaction coordinates, locates the dividing surface (often the $q\approx0.5$ isosurface), and characterises the structure of the transition-state ensemble \cite{vanden2006towards,Kang2024,he2022,Chen2023}. Besides, because the committor encodes both kinetic and thermodynamic information, its accurate estimation enables quantitative prediction of kinetics \cite{vanden2006towards}, precise interpretation of maximum-flux paths and transition tubes \cite{vanden2006towards,Berkowitz1983}, and principled coarse-graining of high-dimensional molecular and stochastic systems \cite{Kang2024}.

There are two common ways to estimate the committor function in the general case. The first approach uses Monte Carlo sampling based on the definition \eqref{definition_committor}, which is often called committor analysis in the transition-path sampling literature~\cite{Dellago2006,jung2023,Bolhuis2002}. It requires many independent simulations from each configuration of interest until hitting one of the two metastable sets, so it becomes prohibitively expensive in practice. 
The second approach considers the backward Kolmogorov equation for committor function~\cite{vanden2010transition}:
\begin{equation}\label{general_committor}
\left\{
\begin{aligned}
    &\LLL q = 0\\
    &q(\xx)\big|_{\pp A} = 0\\
    &q(\xx)\big|_{\pp B} = 1
\end{aligned}
\right.
\end{equation}
where \(\LLL\) is the infinitesimal generator of \eqref{general_SDE}. However, this PDE lives on a high-dimensional domain and therefore suffers from the curse of dimensionality if using traditional methods.

The most interesting and well-studied model for molecular systems is the over-damped Langevin dynamics 
\begin{equation}
    \label{overdamped_SDE}
    \dd \XX_t = -\nabla_{\xx} U(\XX_t) \dt + \sqrt{2\beta^{-1}}\dd \bm{W}_t
\end{equation}
where $U: \Omega \to \real$ is the potential function in configuration space, and $\beta = (k_B T)^{-1}$ is the inverse temperature. In this setting, the infinitesimal generator is self-adjoint on the weighted Sobolev space $H^1_{\mu}(\Omega)$, and equation \eqref{general_committor} admits the variational form \cite{vanden2006towards,khoo2019solving}
\begin{equation}\label{variational_form_overdamped_langevin}
    \left\{
    \begin{aligned}
        &q_0 = \mathop{\arg\min}\limits_{f\in H^{1}_{\mu}(\real^{n})} \big\{ \ I_0[f]:= \beta^{-1} \int \mu \Vert \nabla_{\xx} f \Vert^2\dd {\xx} \big\} \\
        &\text{s.t.}\quad f|_{\pp A} = 0,\ f|_{\pp B} = 1
    \end{aligned}\right.
\end{equation}
where $\mu \propto \exp(-\beta U)$ is the invariant measure of the system. The subscript $0$ denotes the over-damped Langevin dynamics.
Many machine learning studies solve this variational problem using the Deep Ritz method~\cite{khoo2019solving,lilinren2019,Kranthi2022,lin2024deep,Kang2024,Trizio2025}, and the optimal value $I_0[q_0]$ gives the transition rate between two metastable sets~\cite{vanden2006towards,Kang2024}.

In this work, we aim at computing the  committor function for the following more fundamental system, the under-damped Langevin dynamics: 
\begin{equation}\label{underdampedlangevindynamics1}
    \left\{
    \begin{aligned}
        & \dd \XX_t = \VV_t\dt \\
        & \dd \VV_t = -\nabla_{\xx} U(\XX_t) \dt -\gamma \VV_t \dt + \sqrt{2\gamma \beta^{-1}}\dd \bm{W}_t
    \end{aligned}\right.
\end{equation}
where $\gamma$ is the friction coefficient. As $\gamma \to \infty$, the 
over-damped Langevin dynamics \eqref{overdamped_SDE} arises as the high-friction limit of the under-damped Langevin dynamics (see \cite{Vanden-Eijnden2006,Huisinga2001,Schütte2000}). The committor function \(q=q(\xx,\vv)\) of \eqref{underdampedlangevindynamics1} admits an asymptotic expansion as \(\gamma\to\infty\) \cite{Vanden-Eijnden2006}:
\begin{equation}\label{asymptotic}
    q(\xx,\vv) = q_0(\xx) + \gamma^{-1}\,\vv\cdot\nabla_{\xx} q_0(\xx) + O(\gamma^{-2})
\end{equation}
where \(q_0(\xx)\) is the committor for the over-damped Langevin dynamics with the same potential and temperature. This expansion makes explicit the connection between the two regimes: for large friction \(\gamma\), the over-damped committor \(q_0\) is the leading-order approximation to the full under-damped committor; for moderate or small \(\gamma\), inertial (velocity-dependent) corrections become significant.

Beyond the committor, other important differences exist between the two systems. For example, choosing an appropriate friction matrix can optimize convergence rates according to the central-limit theorems for under-damped dynamics\cite{Lelièvre2023,Kauzlarić2011}. Dynamical and thermodynamic quantities—including entropy production \cite{Bo2012,Bo2014,MATSUO2000188}, transition rates between metastable states \cite{vanden2006towards,vanden2010transition,Vanden-Eijnden2006}, and relations given by thermodynamic uncertainty principles \cite{Lee2019}—also differ substantially from their over-damped counterparts. Consequently, although over-damped Langevin dynamics is a useful simplification, it omits features relevant to many applications. To our knowledge, work on committor functions for under-damped Langevin dynamics remains limited. This raises two natural questions: how can one compute the committor in under-damped systems, and how does the friction coefficient influence transitions between metastable states?

 As shown in Sec.~\ref{Main_Method}, the infinitesimal generator \(\LLL\) of eq.\eqref{underdampedlangevindynamics1} is not self-adjoint with respect to the invariant measure; specifically, \(\LLL\) is a summation of a self-adjoint operator and a skew-adjoint operator. Consequently, the backward Kolmogorov equation cannot be cast as a variational problem with an energy functional depending only on first-order derivatives and zeroth-order terms and the Deep Ritz method is not directly applicable. While PINNs can be used to solve \eqref{general_committor}, the infinitesimal operator \(\LLL\) typically has several eigenvalues near zero \cite{Schütte2000}; consequently, a low PINN loss does not guarantee an accurate committor, since any eigenfunction with an eigenvalue close to zero can yield a small loss. As we demonstrate numerically, PINN solutions may fail to capture the dominant saddle-point structure of the transition. 

In this work, we design an iterative scheme to solve the backward Kolmogorov equation \eqref{general_committor}. Building upon the classical Concus-Golub-Widlund (CGW) splitting method~\cite{concus1976,Widlund1978}, we propose a generalized CGW formulation. By explicitly incorporating a spatial operator derived from the infinitesimal generator of over-damped Langevin dynamics, our generalized scheme overcomes the limitations of standard scalar shifts, achieving a substantially accelerated convergence rate and enhanced algorithmic robustness. Furthermore, drawing on recent variational approaches \cite{khoo2019solving,Kang2024}, we cast the resulting subproblems as optimization tasks that can be efficiently solved by the Deep Ritz method. To ensure high fidelity in our neural network approximations, we introduce an auxiliary potential $\widetilde{U}$ to drive enhanced sampling within the transition region—the area where the committor function changes rapidly. Ultimately, we demonstrate that the proposed method yields highly accurate committor estimates and successfully recovers the sharp isosurface near $q=0.5$, thereby precisely identifying the system’s transition states.

In summary, the main contributions of this work are:

\begin{itemize}
    \item We develop an iterative neural solver of Committor Function in under-damped Langevin dynamics. Formulated as a generalized CGW splitting, our approach generates variational subproblems efficiently solved by the Deep Ritz method. Our method proves highly robust in challenging, high-dimensional regimes, successfully resolving sharp transition interfaces. To our knowledge, this is the first practical numerical solver for under-damped Langevin dynamics.
    \item We introduce a novel auxiliary potential, $\widetilde{U}$, to drive sampling directly toward transition regions. Crucially, this process is distinct from importance sampling, as it does not involve reweighting samples to recover equilibrium distribution. Instead, it acts as a preconditioning operator during training: concentrating training data at the steepest committor gradients substantially accelerates convergence and dramatically improves the recovery of the $q \approx 0.5$ isosurface.
    \item We propose a novel reduction framework for committor function of under-damped Langevin dynamics, establishing a theoretical foundation which proves that the true committor function can be systematically approximated using only collective variables (CVs) and their associated phase-space velocities.
    \item We explore the difference between the over-damped and under-damped Langevin committor functions and transition rates, demonstrating that the friction coefficient, $\gamma$, critically dictates transition kinetics and simulation efficiency.
\end{itemize}

\subsection{Related works}

The first neural-network solver for the committor of over-damped Langevin dynamics was proposed by Khoo \textit{et al.} \cite{khoo2019solving}. They used the variational formulation \eqref{variational_form_overdamped_langevin} and approximated the committor with a neural network trained by minimizing the associated functional. But the integration error of the variational functional is high due to little overlap between the integrand and the invariant measure, resulting in a poor estimation of the committor function.

To overcome the issue of limited data in transition region, the classical enhanced sampling methods are incorperated to generate sufficient samples in transition states. Ren's group proposed the use of classical umbrella sampling and metadynamics to enhance the sampling of the transition region \cite{lilinren2019,lin2024deep,Kranthi2022}.  
Parrinello et~al.\ utilized zero-variance importance sampling to estimate the variational functional \eqref{variational_form_overdamped_langevin} \cite{Kang2024,Trizio2025,kang2025committorsdescriptors}. They presented an iterative scheme and ran biased simulations with the biased potential $-\beta^{-1}\log\Vert \nabla_{\xx} {q_0^k} \Vert^2$, where $q_0^k$ is the learned committor function in the previous step. Recently Wang et al.\ followed the same idea of biased simulation and used a normalizing-flow-based sampler to sample the biased system~\cite{WANG2026114439}.

Other notable approaches for computing the committor in over-damped Langevin systems include a time‑lag variational formulation \cite{he2022,giuseppechen2026,Chen2023,Megías2025}, a semi-group method \cite{Li2022}. Beyond neural networks, alternative machine learning models have been used, for example, a tensor‑network–based method  \cite{chen2023committor}, a Gaussian kernel method\cite{kernelcommittor}, and the Finite Expression Method \cite{Song2025}.

\section{Main Method\label{Main_Method}}
\subsection{Committor function of the under-damped Langevin dynamics}

Consider the under-damped Langevin dynamics
\begin{equation}\label{underdampedlangevindynamics}
    \left\{
    \begin{aligned}
        & \dd \XX_t = \VV_t\dt \\
        & \dd \VV_t = -\nabla_{\xx} U(\XX_t) \dt -\gamma \VV_t \dt + \sqrt{2\gamma \beta^{-1}}\dd \bm{W}_t
    \end{aligned}\right.
\end{equation}
where $U: \Omega \to \real$ is the potential function in configuration space, $\gamma$ is the friction coefficient, $\beta = (\kbt)^{-1}$ is the inverse temperature, and $\bm{W}_t$ is the standard Wiener process in $\mathbb{R}^n$. 
The infinitesimal generator of this diffusion process is
\begin{equation}\label{LLL}
    \LLL u = \vv\cdot \nabla_{\xx} u - (\nabla_{\xx} U+\gamma \vv)\cdot \nabla_{\vv} u + \gamma \beta^{-1} \Delta_{\vv} u
\end{equation}
where $\nabla_{\xx} u = \sum_{i=1}^n\frac{\pp u}{\pp x_i}$, $\nabla_{\vv} u = \sum_{i=1}^n \frac{\pp u}{\pp v_i}$, and $\Delta_{\vv} u = \sum_{i=1}^n\frac{\pp^2 u}{\pp v_i^2}$.
The equilibrium distribution is given by 
\begin{equation}
    \mu(\xx)\otimes \pi(\vv) = \frac{1}{Z_U} e^{-\beta U} \times \frac{1}{Z_{\vv}}e^{-\frac{\beta}{2}\nnorm{\vv}^2} = \frac{1}{Z_H}e^{-\beta H(\xx,\vv)}
\end{equation}
where $H(\xx,\vv) = U(\xx) + \frac{1}{2}\nnorm{\vv}^2$ is the Hamiltonian of the system and $Z_H = \iint_{\Omega \times \mathbb{R}^n }e^{-\beta H(\xx,\vv)}\dd \xx \dd \vv$ is the partition function.

Given two metastable sets $A,B\subset \Omega$, we focus on solving the backward Kolmogorov equation with Dirichlet boundary conditions:
\begin{equation}
\left\{
\begin{aligned}
    &\LLL q  = 0 \\
    &q\vert_{\pp A \times \mathbb{R}^n} = 0, \ \ \ q\vert_{\pp B \times \mathbb{R}^n}= 1
\end{aligned}
\right.
    \label{committor}
\end{equation}
We present a decomposition theorem for the operator  $\LLL$. This implies that we cannot cast equation \eqref{committor} as a variational problem with an energy functional depending only on first-order derivatives and zeroth-order terms. Consequently, the Deep Ritz method—which applies directly in the over-damped case—does not extend straightforwardly to the committor in the under-damped case.

\begin{theorem}
Let $\LLL$ , defined in \eqref{LLL}, be an operator acting on the Hilbert space $H^1_{\mu\otimes\pi}(\Omega\times\mathbb{R}^n)$. 
The operator $\LLL$ admits the decomposition
\begin{equation}\label{decomposition}\LLL = \AAA + \gamma \SSS\end{equation}
where, for any sufficiently smooth function $u=u(\xx,\vv)$,
the component operators are defined as:
\begin{equation}\label{A_and_S}
\AAA u := \vv\cdot \nabla_{\xx} u - \nabla_{\xx} U\cdot \nabla_{\vv} u 
\quad \text{and} 
\quad \SSS u := -\vv\cdot \nabla_{\vv} u + \beta^{-1} \Delta_{\vv} u
\end{equation}
Moreover, with respect to the inner product of the weighted space $L^2_{\mu\otimes\pi}(\Omega\times\mathbb{R}^n)$, the operator $\AAA$ is skew-adjoint, while $\SSS$ is self-adjoint and negative semi-definite, i.e., the adjoint operators $\AAA^*$, $\SSS^*$ satisfy:
\begin{equation}
   \AAA^* = -\AAA 
\quad \text{and} 
\quad \SSS^* = \SSS 
\end{equation}
In particular, $\SSS$ acts exclusively on the velocity variables and corresponds exactly to the infinitesimal generator of the following Ornstein–Uhlenbeck process:
\begin{equation} \label{OU-V}
 \dd \VV_t = -\VV_t \dt + \sqrt{2 \beta^{-1}}\dd \bm{W}_t
\end{equation}
\end{theorem}

\begin{proof}
    For sufficiently smooth test functions $f=f(\xx,\vv)$ and $g=g(\xx,\vv)$, 
    since $\mu \otimes \pi =Z_H^{-1} e^{-\beta H(\xx,\vv)}$ with $H(\xx,\vv) = U(\xx) + \frac{1}{2}\nnorm{\vv}^2$, we calculate by integration-by-part
\begin{equation}
\begin{aligned}
    \langle f,\AAA g \rangle_{\mu\otimes \pi} &= \frac{1}{Z_H}\iint_{\Omega\times \mathbb{R}^{n}} f(\xx,\vv)\,\left(\vv\cdot \nabla_{\xx} g(\xx,\vv) - \nabla_{\xx} U(\xx)\cdot \nabla_{\vv}g(\xx,\vv) \right)e^{-\beta H(\xx,\vv)}\dd \xx \dd\vv \\
    &= -\frac{1}{Z_H}\iint_{\Omega\times \mathbb{R}^{n}} \left( \nabla_{\xx} \cdot(\vv fe^{-\beta H}) - \nabla_{\vv}\cdot ( \nabla_{\xx}U f e^{-\beta H}) \right)g\,\dd \xx \dd\vv \\
    &= -\frac{1}{Z_H}\iint_{\Omega\times \mathbb{R}^{n}} \left( \vv\cdot\nabla_{\xx}
    (fe^{-\beta H}) -  \nabla_{\xx}U\cdot    \nabla_{\vv} (fe^{-\beta H} ) \right) g\, \dd \xx \dd\vv \\
    &= -\frac{1}{Z_H}\iint_{\Omega\times \mathbb{R}^{n}} \left( \vv\cdot \big(\nabla_{\xx}
    f -\beta f \nabla_{\xx} U\big) -  \nabla_{\xx} U \cdot  \big(    \nabla_{\vv} f -\beta f \vv  \big) \right) g\, e^{-\beta H}\, \dd \xx \dd\vv 
    \\
    &= -\langle \AAA f, g \rangle_{\mu\otimes \pi}
\end{aligned}
\end{equation}
The conclusion about 
$\SSS$ is obvious since it is the infinitesimal operator of \eqref{OU-V}. Specifically, 
by integration-by-part for $\vv$ only, we have
\begin{equation}
\begin{aligned}
    \langle f,\SSS g \rangle_{\mu\otimes \pi} &= \frac{1}{Z_H}\int_{\Omega}\left(\int_{ \mathbb{R}^{n}} f(\xx,\vv) \left(-\vv\cdot \nabla_{\vv} g(\xx,\vv) + \beta^{-1} \Delta_{\vv} g(\xx,\vv)\right) \,e^{-\beta H(\xx,\vv)}\dd \vv \right) \dd\xx \\
    &= \frac{1}{Z_H}\int_{\Omega}\left(\int_{ \mathbb{R}^{n}} f\,  \left(-\nabla_\vv H \cdot \nabla_{\vv} g + \beta^{-1} \Delta_{\vv} g\right) \,e^{-\beta H}\dd \vv \right) \dd\xx 
    \\
    &= \frac{1}{\beta Z_H}\int_{\Omega}\left(\int_{ \mathbb{R}^{n}} f \, \nabla_{\vv}\cdot \left(\nabla_{\vv}g\ e^{-\beta H}\right)\, \dd \vv \right) \dd\xx \\
    &= -\frac{1}{\beta Z_H}\int_{\Omega}\left(\int_{ \mathbb{R}^{n}}\nabla_{\vv} f \cdot \nabla_{\vv}g\ e^{-\beta H}\dd \vv \right) \dd\xx \\
    &= \frac{1}{\beta Z_H}\int_{\Omega}\left(\int_{ \mathbb{R}^{n}}g\, \nabla_{\vv}\cdot (\nabla_{\vv}f\ e^{-\beta H})\dd \vv \right) \dd\xx \\
    &= \langle \SSS f, g \rangle_{\mu\otimes \pi}
\end{aligned}
\end{equation}
\end{proof}

\subsection{Iterative Neural Variational Scheme}

Since the infinitesimal generator $\LLL$ admits the decomposition \eqref{decomposition},  
it is natural to consider the following  iterative scheme, which belongs to the  
CGW splitting method for large non-symmetric linear systems \cite{concus1976,Widlund1978}:
\begin{equation}\label{itr_naive}
     \left\{
     \begin{aligned}
         & (\omega_1 \III-\gamma \SSS) q^{k+1} = (\omega_1 \III + \AAA) q^{k}\\
         & q^{k+1}(\xx,\vv) = 0 ,\ \ \forall \xx\in \pp A,\ \vv \in \real^{n}\\
         & q^{k+1}(\xx,\vv) = 1 ,\ \ \forall \xx\in \pp B,\ \vv \in \real^{n}
     \end{aligned}
     \right.
 \end{equation}
 where $\omega_1>0$ is a hyperparameter serving as regularization.
 The traditional CGW for linear systems adopts the 
 generalized Conjugate Gradient  algorithm by utilizing the symmetry of $\SSS$. Here for the infinite dimensional problem,  this  update \(q^{k+1}\) could be computed  in the following variational way since \(\SSS\) is negative semi-definite with respect to the weighted inner product
$\pi$:
\begin{equation}\label{variational_form_itr_naive}
\left\{
    \begin{aligned}
        &q^{k+1} = \mathop{\arg\min}\limits_{f\in H^{1}_{\mu\otimes\pi}(\Omega\times \real^{n})} \ \frac{1}{2} \beta^{-1} \langle \nabla_{\vv} f,\nabla_{\vv} f \rangle_{\mu\otimes \pi} + \omega_1\langle f,f \rangle_{\mu\otimes\pi} - \langle \omega_1 q^k + \AAA q^{k}, f \rangle_{\mu\otimes \pi} \\
        &\text{s.t.} ~ \ f|_{\pp A \times \real^{n}} = 0,\ f|_{\pp B\times \real^{n}} = 1
    \end{aligned}
\right.
\end{equation}

The primary bottleneck of the CGW method is the optimal selection of the shift parameter $\omega_1$~\cite{Benzi2022}. Choosing an $\omega_1$ that is too small causes the iteration operator's spectral radius to exceed $1$ and diverge, while an excessively large value forces it to approach $1$ from below, severely stagnating the algorithm. Navigating this parameter sensitivity is exceptionally difficult because the theoretically optimal $\omega_1$ depends strictly on the extreme eigenvalues of both the Hermitian and skew-Hermitian operators. In the context of large-scale, high-dimensional systems, computing these exact spectral bounds is more computationally prohibitive than solving the linear system itself.

To overcome this fundamental bottleneck, we generalize the classic CGW framework by replacing the standard scalar shift with an operator-valued formulation. Specifically, we consider the following iterative scheme:
\begin{equation}\label{itr}
\left\{
\begin{aligned}
&\KKK q^{k+1} - \gamma \SSS q^{k+1} = \KKK q^{k} +\AAA q^{k}\\
& q^{k+1}(\xx,\vv) = 0 ,\ \ \forall \xx\in \pp A,\ \vv \in \real^{n}\\
& q^{k+1}(\xx,\vv) = 1 ,\ \ \forall \xx\in \pp B,\ \vv \in \real^{n}
\end{aligned}
\right.
\end{equation}
Rather than relying on a simple scalar multiple of the identity, we design the shift operator $\KKK$ to incorporate macroscopic spatial relaxation dynamics, defined as:
\begin{equation}\label{KKK}
\KKK=\omega_{1}\III-\omega_{2}\KKK_{\xx}
\end{equation}
where $\omega_{1},\omega_{2}>0$ serve as tuning parameters, and $\III$ denotes the identity operator. Furthermore, $\KKK_{\xx}$ is the spatial operator motivated by the over-damped Langevin generator associated with an auxiliary potential $\widetilde{U}(\xx)$, given by:
\begin{equation}\label{kx}
\KKK_{\xx}f=-\nabla_{\xx}\widetilde{U}\cdot\nabla_{\xx}f+\beta^{-1}\,\Delta_{\xx}f
\end{equation}
 With this formulation, the combined operator $\KKK-\gamma\SSS$ is coercive and positive definite under the weighted product measure
\begin{equation}\label{tilde_mu_otimes_pi}
\widetilde{\mu}\otimes\pi\propto \exp\!\Big(-\beta\big(\widetilde{U}(\xx)+\tfrac{1}{2}\lVert \vv\rVert^{2}\big)\Big)
\end{equation}
and we can derive an equivalent variational formulation for the subproblem \eqref{itr}.

\begin{theorem}
Let $\omega_1, \omega_2 > 0$ be given parameters in \eqref{KKK}, and let $\tilde{\mu}$ and $\pi$ be the measures as defined in \eqref{tilde_mu_otimes_pi}. Define the following two bilinear forms:
\begin{equation}
\BBB_2(f,g) =\ \  \omega_1 \langle f,g \rangle_{\tilde{\mu}\otimes \pi} + \gamma\beta^{-1}\langle \nabla_{\vv} f,\nabla_{\vv} g \rangle_{\tilde{\mu}\otimes \pi}\ \ +\omega_2\beta^{-1}\langle \nabla_{\xx} f,\nabla_{\xx} g \rangle_{\tilde{\mu}\otimes \pi}
\label{BK}
\end{equation}
\begin{equation}
\BBB_1(f,g) = \langle  \AAA f,g \rangle_{\tilde{\mu}\otimes \pi}             -\gamma\beta^{-1}\langle  \nabla_{\vv} f,\nabla_{\vv} g \rangle_{\tilde{\mu}\otimes \pi}
\end{equation}
Then, $q^{k+1}$ solves the operator equation~\eqref{itr} if and only if it is the unique minimizer of the following optimization problem:
\begin{equation}\label{variational}
\left\{
\begin{aligned}
&q^{k+1}=\mathop{\arg\min}_{f\in H^{1}_{\widetilde{\mu}\otimes\pi}(\Omega\times\mathbb{R}^{n})}
\;\tfrac{1}{2}\,\BBB_{2}(f-q^{k},f-q^{k})-\BBB_{1}(q^{k},f)\\
&\text{s.t.}\quad f|_{\partial A\times \real^{n}}=0,\; f|_{\partial B\times \real^{n}}=1
\end{aligned}
\right.
\end{equation}
\end{theorem}

\begin{proof}
Suppose that $q^{k+1}$ is a classical solution to the operator equation~\eqref{itr}. We introduce the variational functional:
\begin{equation}
I_k[f] := \frac{1}{2}\BBB_2(f-q^k,f-q^k) - \BBB_1(q^k,f)
\end{equation}

Let $\widetilde{\Omega} := \Omega\setminus (A\cup B)$. 
For any function $f,g \in H_{0,\tilde{\mu}\otimes\pi}^1(\widetilde{\Omega}\times\mathbb{R}^n)$, we first recall the definitions of $\KKK_x$ and $\SSS$ from equations~\eqref{kx} and~\eqref{A_and_S} and apply integration-by-parts to obtain:
\begin{equation}
\langle \nabla_{\xx} f, \nabla_{\xx} g \rangle_{\tilde{\mu}\otimes\pi} = -\beta\langle \KKK_x f, g \rangle_{\tilde{\mu}\otimes\pi}
\label{eq:ibp-x}
\end{equation}
and
\begin{equation}
\langle \nabla_{\vv} f, \nabla_{\vv} g \rangle_{\tilde{\mu}\otimes\pi} = -\beta\langle \SSS f, g \rangle_{\tilde{\mu}\otimes\pi}
\label{eq:ibp-v}
\end{equation}

For any test function $u \in H_{0,\tilde{\mu}\otimes\pi}^1(\widetilde{\Omega}\times\mathbb{R}^n)$, we now compute the Gâteaux directional derivative of $I_k$ at $q^{k+1}$ in the direction $u$:
\begin{equation}
\begin{aligned}
\frac{\dd}{\dd \epsilon}\bigg|_{\epsilon=0} I_k[q^{k+1}+\epsilon u] 
&= \BBB_2(q^{k+1}-q^k, u) - \BBB_1(q^k, u)\\
&= \omega_1 \langle q^{k+1}-q^k, u \rangle_{\tilde{\mu}\otimes \pi} + \gamma\beta^{-1}\langle \nabla_{\vv}(q^{k+1}-q^k), \nabla_{\vv} u \rangle_{\tilde{\mu}\otimes \pi}\\
&\quad + \omega_2\beta^{-1}\langle \nabla_{\xx}(q^{k+1}-q^k), \nabla_{\xx} u \rangle_{\tilde{\mu}\otimes \pi}\\
&\quad - \langle \AAA q^k, u \rangle_{\tilde{\mu}\otimes \pi} + \gamma\beta^{-1}\langle \nabla_{\vv} q^k, \nabla_{\vv} u \rangle_{\tilde{\mu}\otimes \pi}
\end{aligned}
\end{equation}

Applying the integration-by-parts formulas~\eqref{eq:ibp-x} and~\eqref{eq:ibp-v}, we obtain:
\begin{equation}
\begin{aligned}
\frac{\dd}{\dd \epsilon}\bigg|_{\epsilon=0} I_k[q^{k+1}+\epsilon u] 
&= \int_{\widetilde{\Omega}\times\mathbb{R}^n} u\left[\omega_1 (q^{k+1}-q^k) - \gamma \SSS(q^{k+1}-q^k) \right.\\
&\quad\left. - \omega_2\KKK_x(q^{k+1}-q^k) - \AAA q^k - \gamma \SSS q^k\right] (\tilde{\mu}\otimes\pi)  \, \dd x \, \dd v\\
&= -\int_{\widetilde{\Omega}\times\mathbb{R}^n}  u \left(\KKK q^{k+1} - \gamma \SSS q^{k+1} - \KKK q^k - \AAA q^k\right) (\tilde{\mu}\otimes\pi) \, \dd x \, \dd v
\end{aligned}
\end{equation}

Since $q^{k+1}$ satisfies equation~\eqref{itr}, the expression in parentheses vanishes for all $u$, and thus
\begin{equation}
\frac{\dd}{\dd \epsilon}\bigg|_{\epsilon=0} I_k[q^{k+1}+\epsilon u] = 0 \quad \text{for all admissible } u
\end{equation}

This shows that $q^{k+1}$ is a critical point of $I_k[\cdot]$ subject to the boundary constraints in problem~\eqref{variational}. Since $\BBB_2(\cdot,\cdot)$ is strictly positive definite and coercive, the functional $I_k$ is strictly convex. Therefore, $q^{k+1}$ is the unique minimizer of the optimization problem~\eqref{variational}.
\end{proof}

The variational problem~\eqref{variational} is well-posed and amenable to numerical solution via the Deep Ritz method. Specifically, the expectations within the bilinear forms can be efficiently approximated using Monte Carlo estimation, while the Dirichlet boundary conditions are weakly enforced via a penalty method. The boundary penalty functional is defined as:\begin{equation}L_b[f] = \frac{\lambda_b}{N_A}\sum_{\xx_i\in \pp A,\vv_i\in \real^n} |f(\xx_i,\vv_i)|^2 + \frac{\lambda_b}{N_B}\sum_{\xx_i\in \pp B,\vv_i\in \real^n} |f(\xx_i,\vv_i)-1|^2\end{equation}where $\lambda_b>0$ is the penalty weight, and $N_A$, $N_B$ denote the number of collocation points sampled on the boundaries $\pp A$ and $\pp B$, respectively. For a neural network ansatz $f$ parameterizing the candidate solution for the next iterate, the total empirical loss function to be minimized is formulated as:\begin{equation}L[f] = \tfrac{1}{2}\,\BBB_{2}(f-q^{k},f-q^{k})-\BBB_{1}(q^{k},f) + L_b[f]\end{equation}

We emphasize that the auxiliary potential $\widetilde{U}$ and the tuning parameters $\omega_1, \omega_2$ serve two critical mathematical and computational purposes.
First, by explicitly embedding the spatial $\xx$-dependence into the shift operator $\KKK$, the generalized CGW framework achieves a significantly accelerated convergence rate and substantially greater robustness in parameter selection.
Second, the introduction of $\widetilde{U}$ establishes a flexible framework for enhanced sampling. By carefully designing $\widetilde{U}$ alongside the relative weights $\omega_1$ and $\omega_2$, one can bias the Monte Carlo sampling toward transition regions in the configurational space that are most vital for an accurate committor approximation. The explicit construction of $\widetilde{U}$ and practical heuristics for parameter selection are thoroughly detailed in Section~\ref{tilde-U}.

Remark. The Hermitian and Skew-Hermitian Splitting (HSS) method\cite{HSS-BGN2023} serves as an unconditionally convergent generalization of the classical CGW scheme. Each iteration consists of two distinct half-steps: one requiring the solution of a self-adjoint operator equation, and the other requiring the solution of a skew-adjoint operator equation. However, directly extending the HSS framework to deep learning applications presents two fundamental mathematical challenges. First, exactly solving the high-dimensional skew-adjoint operator equation is computationally prohibitive when the solution is parameterized by a neural network. Second, the spectrum of the dissipation operator $-\SSS$ is highly degenerate, encompassing $0$ and extending to $+\infty$. This inherent lack of strict coercivity forces the spectral radius of the iteration operator to be exactly $1$, thereby invalidating the method's theoretical guarantee of unconditional convergence. Consequently, the adaptation of the HSS method to this context remains a subject for future work.

\subsection{Enhanced sampling strategy\label{tilde-U}}
 Sampling the rare transition region is a major challenge for committor function calculation, as documented in the literature \cite{lilinren2019,lin2024deep,Kang2024,Trizio2025}. In our work, by allowing \(\widetilde{U}\) to be chosen freely (rather than using the system’s potential), we obtain a flexible strategy to bias sampling toward dynamically important regions and thereby improve learning efficiency.

We consider two distinct strategies for choosing $\widetilde{U}$ and designing adaptive sampling of the associated measure $\tilde{\mu}$:
 First, we bias the potential to place more samples in the transition region where the committor varies rapidly. For example, recent work proposes adding a bias proportional to \(-\beta^{-1}\log\lVert\nabla_{\xx} q\rVert^{2}\) to concentrate sampling where the spatial gradient of the committor is large \cite{Kang2024}. Second, residual‑based adaptive sampling~\cite{WU2023} targets points where the PDE residual is large. Since the committor is nearly constant near the metastable basins \(A\) and \(B\) (yielding small residuals) but changes steeply across the transition region, residual‑driven refinement naturally focuses training data in the most informative region. In summary, this work utilizes the formulation $\widetilde{U}^{k+1} \propto \log\big\{{\beta_{\text{res}}W_{\text{res}} - \beta_{q^k} \|\nabla_\xx q^{k}\| + \beta_{\text{uniform}}}\big\}$, where $W_{\text{res}}$ is the weight function derived from the PDE residue, $q^k$ denotes the committor estimate from the previous iteration, and the constants $\beta_{\text{res}}$, $\beta_{q^k}$, and $\beta_{\text{uniform}}$ are tunable hyperparameters.

A critical practical advantage of our framework is that both bilinear forms $\BBB_2$ and $\BBB_1$ require \emph{no explicit analytical formula} for $\widetilde{U}$. They depend solely on samples drawn from the measure $\tilde{\mu}\otimes\pi$. This decoupling provides substantial flexibility: one can freely design, test, and update sampling strategies without needing to evaluate or differentiate $\widetilde{U}$ analytically.

Furthermore, the product structure of the measure $\tilde{\mu}\otimes\pi$ yields a crucial decomposition: the spatial and velocity variables decouple, with the velocity marginal remaining Gaussian. This decomposition enables inexpensive independent resampling of velocities at each iteration while maintaining concentrated position sampling in the transition region. The benefits are twofold:
\begin{enumerate}
    \item \emph{Generalization:} The network encounters a diverse set of velocity realizations, improving robustness across the velocity space.
    \item \emph{Efficiency:} Costly position sampling (which often requires specialized techniques to target the transition region) can be reused across multiple velocity samples.
\end{enumerate}

In summary, the freedom in choosing $\widetilde{U}$, combined with the separability of $\tilde{\mu}\otimes\pi$, provides a powerful and practical toolkit for overcoming the sampling bottleneck in committor computation. These design principles are realized concretely in Algorithm~\ref{framework} below.

\begin{algorithm}
\caption{Iterative Solver For Committor Function\label{framework}}
\begin{algorithmic}[1]
\State \textbf{Input:} Initial approximation $q^0$, tolerance $\varepsilon > 0$, maximum iterations $K_{\max}$, biased measure $\tilde{\mu} \propto \exp(-\beta \widetilde{U})$
\State Initialize $k \gets 0$, $q^{(k)} \gets q^0$
\While{$k < K_{\max}$}
    \State \textbf{Sampling Phase:}
    \State \quad Sample velocities $\{\vv_i\}_{i=1}^{N_v} \sim \pi(\vv)$ independently
    \State \quad Sample positions $\{\xx_j\}_{j=1}^{N_x} \sim \tilde{\mu}(\xx)$ from the biased measure
    
    \State \textbf{Training Phase:}
    \State \quad Train neural network to minimize:
    \begin{equation*}
        L[q] = \tfrac{1}{2}\BBB_2(q - q^{(k)}, q - q^{(k)}) - \BBB_1(q^{(k)}, q) + L_b[q]
    \end{equation*}
    
    \State \quad Set $q^{(k+1)} \gets \text{optimized network}$
    
    \State \textbf{Convergence Check:}
    \If{$\|q^{(k+1)} - q^{(k)}\|_{L^2(\tilde{\mu}\otimes\pi)} \le \varepsilon$}
        \State \Return $q^{(k+1)}$ 
    \EndIf

    \State \textbf{Optional: Update auxiliary potential}
    \State \quad $\widetilde{U} \gets \text{Updated based on committor gradient}$ 
    
    \State $k \gets k + 1$, $q^{(k)} \gets q^{(k+1)}$

\EndWhile
\State \Return $q^{(k)}$ 
\end{algorithmic}
\end{algorithm}

\subsection{Continuous Limit of the Iterative Scheme}

The iterative scheme~\eqref{itr} admits a natural interpretation as a temporal discretization of a continuous evolution equation for an auxiliary field $\tilde{q}(\xx,\vv,t)$. We regard the stationary PDE \eqref{itr} as the long time limit of the following initial-boundary value problem:
\begin{equation}
     \left\{
     \begin{aligned}
         &\pp_t \tilde{q}(x,v,t) = \widetilde{\KKK}^{-1}\LLL \tilde{q}(x,v,t)\\
         & \tilde{q}(\xx,\vv,t) = 0 ,\ \ \forall \xx\in \pp A,\ \vv \in \real^{n}\\
         &\tilde{q}(\xx,\vv,t) = 1,\ \ \forall \xx\in \pp B,\ \vv \in \real^{n}
     \end{aligned}
     \right.\label{timeevo}
 \end{equation}
where $\LLL = \AAA + \gamma\SSS$ is the infinitesimal generator and $\widetilde{\KKK}$ is a positive-definite, invertible auxiliary operator. Applying the semi-implicit Euler scheme with timestep $\Delta t$ to~\eqref{timeevo} yields:
\begin{equation}
\widetilde{\KKK}\,\tilde{q}^{\,k+1} - \Delta t\, \gamma\, \SSS\, \tilde{q}^{k+1} = \widetilde{\KKK}\,\tilde{q}^{\,k} + \Delta t\,\AAA\,\tilde{q}^{\,k}
\end{equation}
Rearranging and identifying $\KKK := \widetilde{\KKK}/\Delta t$, we recover the discrete iteration~\eqref{itr}. This establishes that the fixed-point iteration~\eqref{itr} provides a time-discrete approximation to the continuous flow generated by $\widetilde{\KKK}^{-1}\LLL$.

A practical question is how to choose the hyperparameters \(\omega_{1},\omega_{2},\widetilde{U}\) (or, equivalently, the operator \(\widetilde{\KKK}\)) to guarantee convergence of the time evolution.  Rather than present a complete spectral analysis, we highlight two key considerations:

\textit{Case 1: Identity auxiliary operator.} The simplest choice $\widetilde{\KKK} = \III$ yields the evolution $\partial_t \tilde{q} = \LLL\tilde{q}$. This evolution is well-posed with long-time asymptotics corresponding to the committor. Consequently, convergence of the discrete scheme follows from standard stability theory for the time integrator (e.g., a sufficiently small timestep $\Delta t$).

\textit{Case 2: Spectral gap condition.} More generally, exponential convergence of the flow generated by $\widetilde{\KKK}^{-1}\LLL$ to its steady state is guaranteed if there exists $r > 0$ such that $\|(\lambda I - \mathcal{L})^{-1}\|$ is uniformly bounded for all $\operatorname{Re}(\lambda) \geq -\gamma$. 

Under this spectral gap condition, the semigroup $e^{t\widetilde{\KKK}^{-1}\LLL}$ decays exponentially with a rate of at least $r$. The time-discrete iteration then converges provided $\Delta t$ respects the stability constraints of the chosen integrator.

Intuitively, larger $\omega_1,\omega_2$ lead to more stable optimization but a lower convergence rate. In practice, numerical parameter tuning guided by operator norm estimates or small-scale spectral computations can help with hyperparameter selection. A fully rigorous automated parameter selection strategy remains an open problem and is left to future investigations.

\subsection{Committor Function in Low-Dimensional Descriptors}\label{function_approximation}

A standard and practically motivated assumption in over-damped Langevin dynamics is that the committor depends primarily on a small set of collective variables (CVs) $s(\xx)$, where $s: \Omega \to \mathbb{R}^d$ with $d \ll n$. This assumption reflects the observation that complex dynamics are often effectively governed by a few slowly-varying degrees of freedom~\cite{Schütte2000,Kang2024,vanden2006towards,liu2026fesfmfreeenergysurface}. Under this assumption, the leading-order committor admits the approximation $q_0(\xx) \approx q_0(s(\xx))$.

Substituting this ansatz into the asymptotic expansion~\eqref{asymptotic} and retaining terms up to $O(\gamma^{-1})$, we obtain:
\begin{equation}\label{asym_CV}
q(\xx,\vv) \approx q_0(s(\xx)) + \gamma^{-1}\vv \cdot J_s(\xx) \nabla_s q_0(s(\xx)) + O(\gamma^{-2})
\end{equation}
where $J_s(\xx) = (\partial s_j/\partial x_i)$ denotes the Jacobian of the CV map. The correction term admits a natural physical interpretation: noting that $\vv = \dot{\xx}$ and $J_s^T \vv = \dot{s}$, equation \eqref{asym_CV} can be rewritten as
\begin{equation}\label{dot_s}
q(\xx,\vv) \approx q_0(s(\xx)) + \gamma^{-1}\dot{s} \cdot \nabla_s q_0(s(\xx)) + O(\gamma^{-2})
\end{equation}
This represents a directional derivative of the coarse-grained committor along the instantaneous CV trajectory, making explicit the coupling between configurational and velocity dependence.

The asymptotic structure in \eqref{dot_s}  motivates a reduced-variable parametrization:
\begin{equation}
q(\xx,\vv) \approx Q(s(\xx), \dot{s}(\xx)) = Q(s(\xx), J_s^T(\xx)\vv)
\label{eq:reduced_param}
\end{equation}
where  $Q:\real^{d}\times\real^{d}\to [0,1]$ is a parametrized neural network.
This captures the most dynamically relevant dependence on configuration and instantaneous coarse-grained velocity. This choice offers three significant advantages. 

\emph{First, dimensionality reduction:} When the dynamics are effectively governed by a few CVs, this parametrization dramatically reduces the function-approximation dimension compared to working in the full $(\xx,\vv)$ space. This reduction substantially improves sample efficiency and simplifies the complexity of neural network or spectral representations, enabling faster convergence and more stable training.

\emph{Second, physical interpretability:} Expressing the ansatz in terms of $s$ and $\dot{s}$ makes the physical structure of the committor explicit. The leading-order term depends only on CV position and encodes the configurational likelihood of reaching $B$ from $A$ and the $O(\gamma^{-1})$ correction accounts for the instantaneous motion along those CVs.

\emph{Third, symmetry preservation:} If the chosen CVs respect physical symmetries (e.g., rotational or translational invariance), the reduced representation automatically inherits those symmetries. This ensures that the learned committor is consistent with the underlying physics, improves generalization beyond the training data, and can significantly reduce the number of parameters needed for an accurate approximation.

The reduced parametrization~\eqref{eq:reduced_param} thus provides both a mathematically principled and physically interpretable framework for efficient committor approximation in high-dimensional systems.

 \section{Numerical Experiments}
 In this section we evaluate the performance of our Iterative Neural Variational Solver on a set of benchmark problems. All neural network models are implemented using the PyTorch framework~\cite{paszke2019pytorch}. We parameterize the approximated solution, $q_\theta$, using a fully connected feedforward architecture. To ensure the network is sufficiently smooth to compute the necessary derivatives for the governing operators, we employ the Softplus activation function across all hidden layers. Furthermore, a Sigmoid activation is applied at the output layer to strictly bound the predicted committor values within the physically meaningful probability interval of $[0, 1]$. The network weights are optimized using the Adam algorithm~\cite{kingma2014adam} to minimize the associated loss function. The specific hyperparameter configurations used for this training process are detailed in \ref{details}.

 \subsection{Double well potential}
 
    Firstly we consider the potential~\cite{khoo2019solving} in $\real^{d}$
    \begin{equation}
        U(x) = (x_1^2-1)^2 + 0.3\sum_{i=2}^d x_i^2
    \end{equation}
    and 
    $$A=\{x \in \mathbb{R}^d: x_1\leq-1\}\ \ \ B=\{x\in \mathbb{R}^d:x_1\geq1\}$$
The true committor function $q(\xx,\vv) \in \real^{d}\times\real^{d}$, while
this problem is effectively one-dimensional in configuration space, and the true committor is effectively two-dimensional in the $(x_1, v_1)$ variables, i.e. $q(\xx,\vv) = q(x_1,v_1)$. While the over-damped case admits an explicit analytical solution for the committor, the under-damped Langevin setting does not. We therefore compute a reference solution using the finite difference method (FDM) and test our iterative neural solver with $\beta = 0.1$ and dimension $d = 3$.

\begin{figure}
\includegraphics[width=\linewidth]{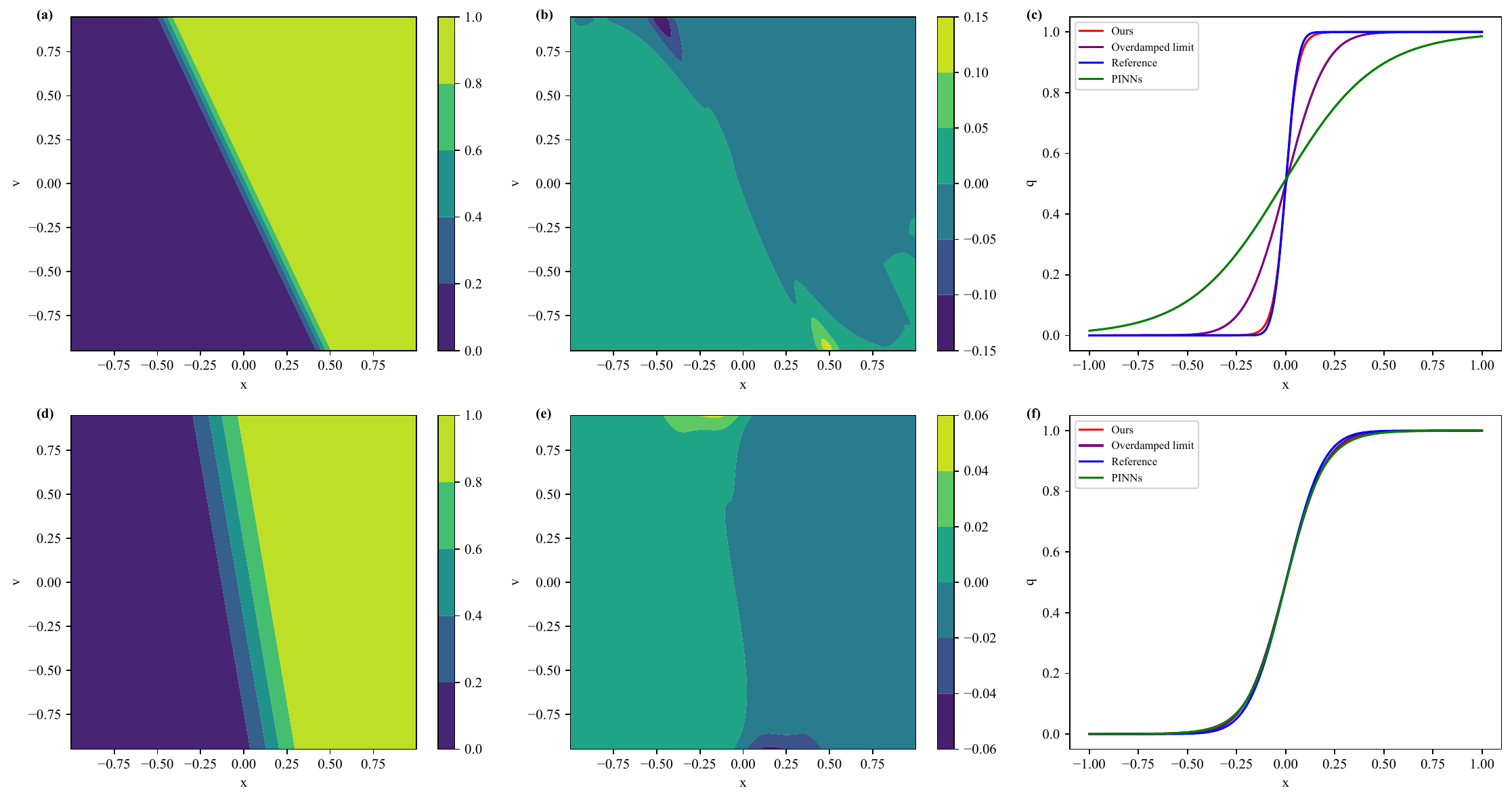}
\caption{\label{1d-results}
Results for the double-well potential with two friction regimes: panels (a--c) correspond to $\gamma = 0.2$ (weakly damped), while panels (d--f) correspond to $\gamma = 5$ (strongly damped). Panels (a) and (d): Learned committor $q_\theta(x_1, v_1)$ with other coordinates fixed. Panels (b) and (e): Pointwise error $Err = q_\theta - q_{\text{ref}}$ between learned and reference solutions. Panels (c) and (f): Committor functions along $x_1$ at a fixed velocity slice $v_1 = 0$, comparing our method with PINN and the FDM reference. Both methods train for 9000 epochs.}
\end{figure}

Fig.~\ref{1d-results}(a) and~\ref{1d-results}(d) display the learned committor for $\gamma = 0.2$ and $\gamma = 5$, respectively, as a function of $(x_1, v_1)$ with remaining coordinates held at fixed values. The corresponding pointwise errors are shown in Fig.~\ref{1d-results}(b) and~\ref{1d-results}(e). A key observation emerges: when $\gamma$ is small, velocity significantly influences the committor; as $\gamma \to \infty$, the committor approaches a function of $x$ only, recovering the over-damped limit. This behavior is in full agreement with the asymptotic analysis of committor asymptotics in under-damped Langevin dynamics.

Fig.~\ref{1d-results}(c) and~\ref{1d-results}(f) compare committors along $x_1$ at $v_1 = 0$ across different methods. In the strongly damped regime ($\gamma = 5$), all methods perform comparably well, and the under-damped committor converges toward the over-damped committor as predicted. However, in the weakly damped regime ($\gamma = 0.2$), PINNs exhibit a substantial loss of accuracy, whereas our method maintains fidelity to the correct solution features. Quantitative error metrics for our iterative neural solver are reported in Table~\ref{1d-errors}.

\begin{table}[H]
    \caption{Error analysis for committor approximation in the double-well potential. $Err = q_\theta - q_{\text{ref}}$ is the difference between the learned neural network solution and a reference solution computed via FDM. Error norms are computed with respect to the measure $\mu\otimes\pi$. }
    \label{1d-errors}
    \centering
    \begin{tabular}{@{} l |l| l l l| l l l @{}}
        \toprule
        $\gamma$ & Method & $\|Err\|_{L_2}$ & $\|Err\|_{L_1}$ & $\|Err\|_{L_\infty}$ & $\|\LLL q_\theta\|_{L_2}$ & $\|\LLL q_\theta\|_{L_1}$ & $\|\LLL q_\theta\|_{L_\infty}$ \\
        \midrule
        0.2 & Ours & \textbf{2.91e-4} & \textbf{2.85e-6} & \textbf{1.60e-1} & \textbf{2.47e-3} & \textbf{1.85e-5} & 1.39 \\
        & PINN & 2.82e-2 & 2.52e-2 & 7.39e-1 & 3.45e-2 & 2.53e-2 & \textbf{7.39e-1} \\
        \midrule
        1 & Ours & \textbf{2.65e-4} & \textbf{2.62e-5} & \textbf{7.94e-2} & \textbf{1.19e-3} & \textbf{8.25e-5} & \textbf{4.10e-1} \\
        & PINN & 3.89e-2 & 1.52e-2 & 9.96e-1 & 5.26e-2 & 1.85e-2 & 9.96e-1 \\
        \midrule
        5 & Ours & \textbf{5.74e-4} & \textbf{2.32e-4} & \textbf{6.04e-2} & \textbf{5.75e-4} & \textbf{1.68e-4} & \textbf{6.04e-2} \\
        & PINN & 6.13e-4 & 2.38e-4 & 6.91e-2 & 6.03e-4 & 1.76e-4 & 6.91e-2 \\
        \midrule
        25 & Ours & \textbf{8.60e-4} & 4.44e-4 & \textbf{5.31e-2} & \textbf{3.12e-4} & 1.33e-4 & \textbf{5.31e-2} \\
        & PINN & 9.35e-4 & \textbf{4.04e-4} & 1.86e-1 & 3.30e-4 & \textbf{1.18e-4} & 1.86e-1 \\
        \midrule
        100 & Ours & \textbf{1.38e-3} & \textbf{7.65e-4} & \textbf{7.89e-2} & \textbf{7.14e-4} & \textbf{3.08e-4} & \textbf{7.89e-2} \\
        & PINN & 9.40e-2 & 9.14e-2 & 3.61e-1 & 2.22e-2 & 1.70e-2 & 3.61e-1 \\
        \bottomrule
    \end{tabular}
\end{table}

Fig.~\ref{rate-estimation} (a) illustrates the convergence of the $L_2$ loss over training iterations under varying hyperparameter configurations. The standard CGW method ($\omega_2 = 0$) exhibits significant sensitivity to the choice of $\omega_1$; specifically, setting $\omega_1 = 0.2$ causes the model to stagnate at a higher loss plateau, whereas increasing it to $\omega_1 = 2$ allows for convergence to a lower loss. In contrast, our method as a generalized CGW  ($\omega_2 = 0.01$) demonstrates highly consistent convergence behavior regardless of the $\omega_1$ value. Both tested configurations of the generalized method ($\omega_1 = 0.2$ and $\omega_1 = 2$) rapidly descend to the same optimal, low-loss regime, matching the best-case performance of the standard CGW method. This result confirms that our method possesses far better robustness in parameter selection, effectively insulating the training process from hyperparameter sensitivity and ensuring stable, accurate convergence.

A notable observation arises when explicitly computing the spectral radius of the iterative operator $(\KKK - \gamma \SSS)^{-1}(\KKK+\AAA)$ using standard finite difference (FDM) discretization. For the four scenarios presented in Fig.~\ref{rate-estimation}, the spectral radii are calculated to be 386.85, 82.58, 7.48, and 5.72, respectively. As these values are significantly greater than one, empirical FDM implementations of our method fail to converge. This indicates that our generalized CGW formulation is only effective within the continuous neural network framework. We attribute this localized convergence to the implicit regularization properties inherent to the neural network optimization process, which effectively dampens the instabilities of the discrete operator.

Given the learned committor $q(\xx, \vv)$, we estimate the transition rate via
\begin{equation}\label{rate_u}
\kappa_{AB} := \beta^{-1}\gamma \int \|\nabla_\vv q(\xx,\vv)\|^2 \, (\tilde{\mu} \otimes \pi)(\dd\xx\, \dd\vv)
\end{equation}
and compare this with the over-damped transition rate $\kappa_0 := I_0[q_0]$ defined in~\eqref{variational_form_overdamped_langevin}. The asymptotic relation~\eqref{asymptotic} predicts
\begin{equation}\label{rate_asmp}
\kappa_{AB} = \gamma^{-1} \kappa_0 + O(\gamma^{-2})
\end{equation}
Fig.~\ref{rate-estimation}(b) reports transition-rate estimates from both the under-damped and over-damped committors. For sufficiently large $\gamma$, the two estimates agree closely, validating the asymptotic prediction. However, they diverge substantially in the weakly damped regime ($\gamma = 0.2$), underscoring the necessity of using the under-damped committor for accurate rate estimation when inertial effects are significant.

\begin{figure}[H]
    \centering
    \includegraphics[width=0.7\linewidth]{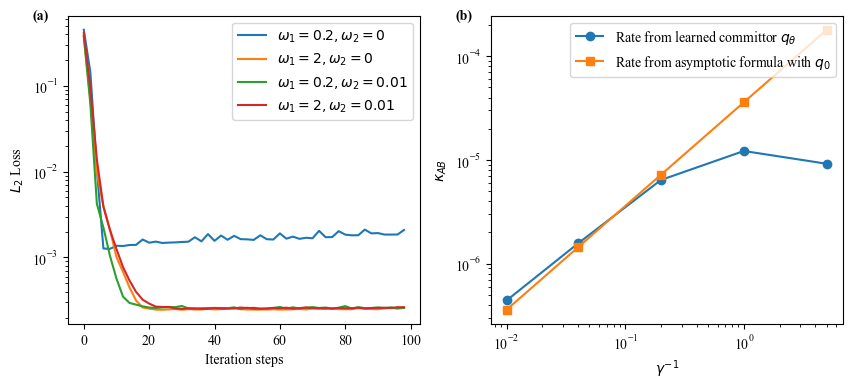}
    \caption{Panel (a): Convergence of the $L_2$ loss during training with varying hyperparameters. Panel (b): Comparison of the transition rate ($\kappa_{AB}$) estimated from the learned committor ($q_\theta$) versus the asymptotic formula ($q_0$) as a function of inverse friction ($\gamma^{-1}$). }
    \label{rate-estimation}
\end{figure}

    \subsection{Müller-Brown potential}
    We next consider the Müller-Brown potential, a prototypical two-dimensional benchmark for testing rare event sampling and transition pathway methods~\cite{khoo2019solving,Kang2024}. The potential is defined as
\begin{equation}
U(x_1,x_2) = \sum_{i=1}^4 A_i \exp\left(a_i(x_1-f_i)^2 + b_i(x_1-f_i)(x_2-g_i) + c_i(x_2-g_i)^2\right)
\end{equation}
where the parameters are specified as follows: $A = [-200, -100, -170, 15]$, $a =  [-1, -1, -6.5, 0.7]$, $b = [0, 0, 11, 0.6]$, $c =  [-10, -10, -6.5, 0.7]$, $f =  [1, 0, -0.5, -1]$, and $g =  [0, 0.5, 1.5, 1]$. 

We focus on transitions between two of the three metastable states, defined as circular regions around their respective minima:
\begin{align}
A &= \left\{(x_1, x_2) \in \mathbb{R}^2 : (x_1 - x_1^A)^2 + (x_2 - x_2^A)^2 \leq r^2\right\} \\
B &= \left\{(x_1, x_2) \in \mathbb{R}^2 : (x_1 - x_1^B)^2 + (x_2 - x_2^B)^2 \leq r^2\right\}
\end{align}
where the minima are located at $[x_1^A, x_2^A] = [-0.55, 1.44]$ and $[x_1^B, x_2^B] = [0.62, 0.028]$ with radius $r = 0.2$. We show the landscape of the potential in Fig.~\ref{muller-brown-potential} and mark the two metastable sets $A,B$ and two important saddle points $S_1,S_2$. The committor level set $\{q = 0.5\}$, which represents the separatrix between $A$ and $B$, is expected to lie predominantly near  $S_1$ due to its larger energy barrier compared to the other saddle point $S_2$. This geometric observation provides a useful heuristic check on the quality of our learned committor functions.

\begin{figure}[ht]
\centering
\includegraphics[width=0.4\linewidth]{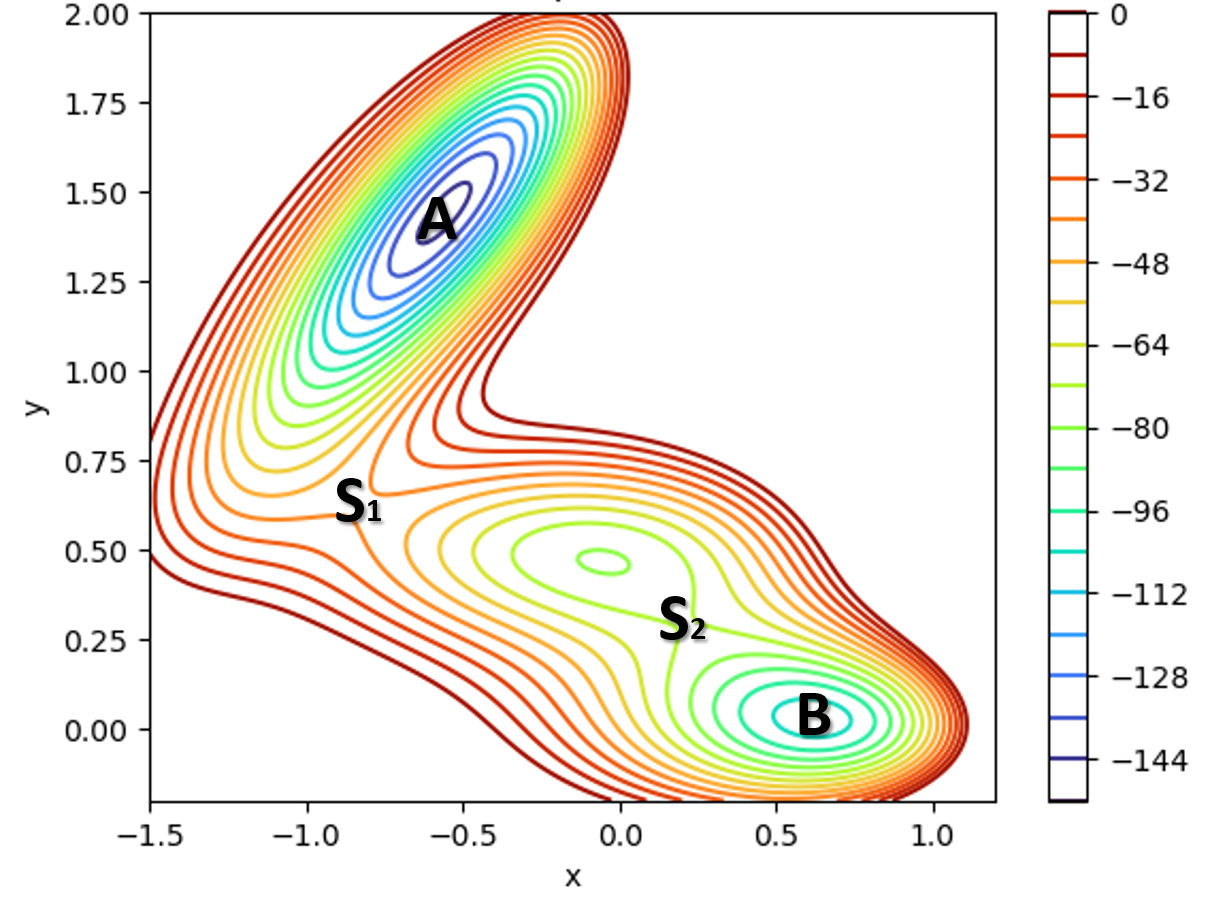}
\caption{Energy landscape of the Müller-Brown potential. The contour plot displays the two target metastable sets $A$ and $B$ (marked with circles) and two critical saddle points $S_1$ and $S_2$ that mediate transitions between basins. }
\label{muller-brown-potential}
\end{figure}

Unlike the over-damped case, the present system has a two‑dimensional configuration space while the committor function depends on four variables, which makes direct numerical solution by classical methods such as finite differences or finite elements prohibitively difficult. To validate our approach  we compute the committor on a uniform grid in configuration space combined with 25 velocity slices, and use these simulations to evaluate the method’s relative error. This sampling strategy allows us to assess performance across both position and velocity dimensions and to obtain robust error statistics.

\begin{figure}[H]
    \includegraphics[width=0.9\linewidth]{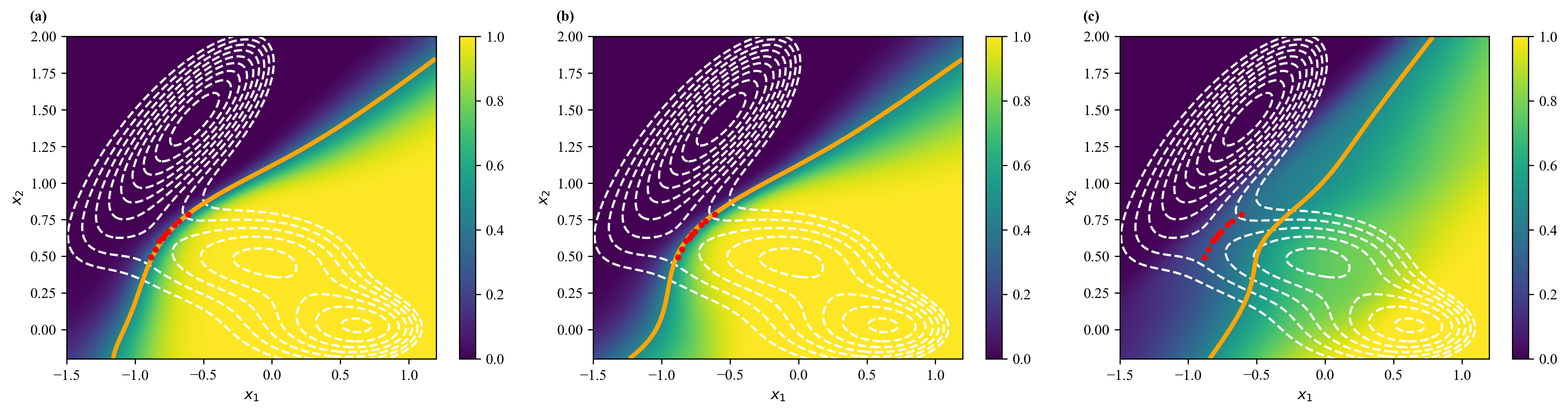}
    \caption{\label{Müller-Brown_results}
    The velocity-averaged committor function at $\kbt=5$ and $\gamma=5$ solved via (a) FDM (reference), (b) our iterative neural solver, and (c) the PINN approach. The solid orange curve indicates the level set $\{q_\theta = 0.5\}$, in comparison to a few red reference points where ($q_{ture} = 0.5$) generated from 1000 unbiased trajectories.}
\end{figure}

    Fig.~\ref{Müller-Brown_results} compares the velocity-averaged committors, $\bar{q}(x)=\langle q(x,\cdot)\rangle_{\pi}$, computed using (a) the FDM reference, (b) our iterative neural solver, and (c) the PINN approach at $\kbt=5$ and $\gamma=5$. As illustrated in panels (a) and (b), the $q_\theta = 0.5$ level set (solid orange line) generated by both the FDM reference and our iterative neural solver passes precisely through the red reference points obtained from unbiased simulations. This precise alignment demonstrates that our method accurately resolves the critical transition boundary and successfully captures the system's true underlying dynamics. Conversely, panel (c) reveals a significant spatial deviation between the PINN-generated level set and the reference points; the PINN solution fails to capture the sharp, localized structure of this boundary. This discrepancy highlights that the PINN method struggles to learn an accurate committor function, specifically failing to resolve the crucial $q \approx 0.5$ region. Ultimately, these results verify that our iterative neural solver yields a higher-fidelity solution, effectively overcoming the training limitations inherent in standard PINN architectures.

\begin{figure}
    \includegraphics[width=0.9\linewidth]{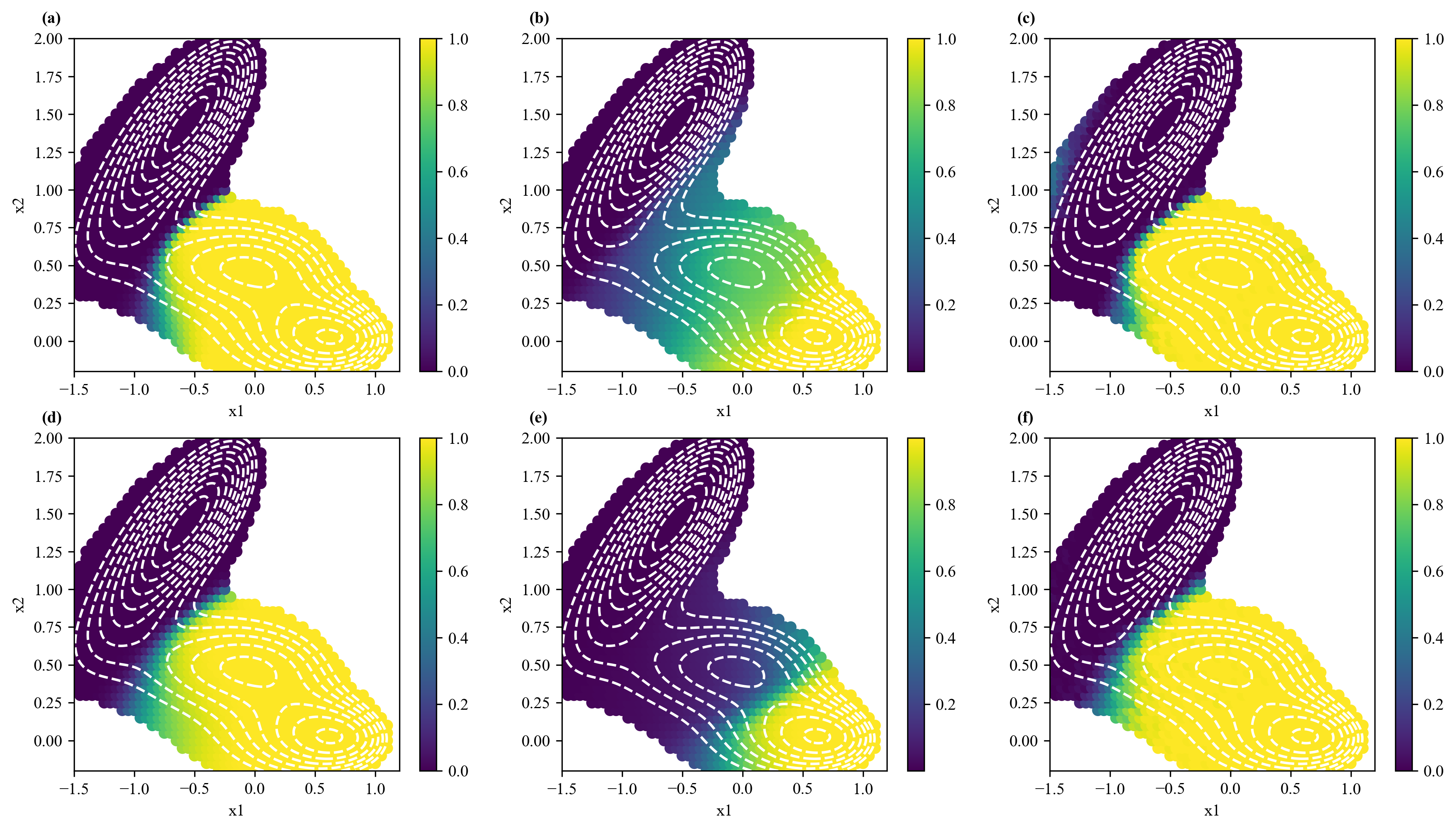}
    \caption{\label{Müller-Brown_gammas}
    Committor function $q_\theta(x_1, x_2)$ at a fixed velocity $(v_1,v_2)$. Results are shown for two friction coefficients: $\gamma = 5$ (top row: a–c) and $\gamma = 25$ (bottom row: d–f). The approximations produced by our iterative neural solver (a, d) and the standard PINN method (b, e) are compared against a reference solution estimated from 1000 unbiased simulation trajectories (c, f). our iterative neural solver trains for 20000 epochs and PINN method trains for 40000 epochs.}
\end{figure}

Fig.~\ref{Müller-Brown_gammas} provides a comprehensive comparison of the committor function $q_\theta(x_1, x_2)$ across the different methods and friction regimes with a fixed $(v_1,v_2)$. As is visually evident for both $\gamma = 5$ and $\gamma = 25$, the solution by our iterative neural solver is nearly identical to the reference true solution obtained from the unbiased simulations. This high fidelity confirms that our approach successfully resolves the correct transition interface across the energy landscape. In contrast, the PINN  solution fails to capture the true transition region; instead, either learns a $q=0.5$ level set at an entirely incorrect location near the secondary saddle point $S_2$ (see Panel (e)), or learns a moderate over saddle points $S_1$ and $S_2$ (see Panel (b)).  From a spectral perspective, a committor-like function whose $q=0.5$ isosurface crosses an incorrect saddle point actually corresponds to an eigenfunction of the governing operator with an eigenvalue very close to zero. This highlights a critical shortcoming of PINNs in complex energy landscapes: because they rely on minimizing a residual, they cannot easily distinguish between the true solution and these near-zero eigenfunctions. Consequently, any linear combination of these eigenfunctions yields a deceptively low PINN loss (see Panels (b) and (e)).  Conversely, the iterative neural solver remains highly robust with varying $\gamma$ and reliably identifies the physically correct transition mechanisms.

Furthermore, a direct comparison between the top ($\gamma = 5$) and bottom ($\gamma = 25$) rows of Fig.~\ref{Müller-Brown_gammas} reveals the distinct effect of the friction coefficient on the spatial structure of the committor function. At $\gamma = 5$, the reactive interface—the intermediate region where the committor probability rapidly shifts from 0 to 1—is relatively sharp. In contrast, increasing the friction to $\gamma = 25$ produces a noticeably broader transition region, reflecting a more diffuse spatial variation across the energy landscape. Crucially, our iterative neural solver accurately captures this friction-dependent broadening, matching the unbiased simulation reference.

\begin{table*}
        \caption{Results for the Müller-Brown potential. Here $Err = q_\theta - q_{ref}$ is the difference between the numerical result and the reference solution. $L_2$ and $L_1$ are calculated under the invariant measure $\mu\otimes \pi$ and $L_2(u)$ is measured under $\mathcal{U}\otimes\pi$ where $\mathcal{U}$ is a uniform distribution over the spatial space. We show that the standard PINN loss $\Vert \LLL q_\theta \Vert_{L_2(u)}$ is not suitable to evaluate the accuracy of the learned committor function.  }
        \centering
        
    \begin{tabular}{@{} l |l| l l l| l l l| l @{}}
\toprule
  $\gamma$ & Method & $\frac{\Vert Err \Vert_{L_2}}{\Vert q_{ref}\Vert_{L_2}}$ & $\frac{\Vert Err \Vert_{L_1}}{\Vert q_{ref}\Vert_{L_1}}$ & $\Vert E rr\Vert_{L_{\infty}}$ & $\Vert \LLL q_\theta \Vert_{L_2}$ & $\Vert \LLL q_\theta \Vert_{L_1}$ & $\Vert \LLL q_\theta \Vert_{L_{\infty}}$& $\Vert \LLL q_\theta \Vert_{L_2(u)}$ \\
\midrule
 1& Ours & \textbf{6.82e-04} & \textbf{1.49e-06} & 1.00e+00 & \textbf{1.20e-03} & \textbf{3.08e-08} & 2.26e+02&5.65e+00 \\
 & PINNs& 1.98e-02 & 3.54e-02 & 1.00e+00 & 2.68e-03 & 6.16e-04 & \textbf{5.61e+00}&2.85e-01 \\
 \midrule
5 & Ours & \textbf{1.16e-03} & \textbf{8.64e-06} & 1.00e+00 & \textbf{1.94e-04} & \textbf{9.66e-09} & 3.66e+01&1.27e+00 \\
& PINNs& 4.72e-02 & 1.19e+00 & \textbf{9.40e-01} & 2.31e-02 & 1.57e-02 & \textbf{2.51e+00}&2.80e-01 \\
\midrule
25 & Ours & \textbf{7.30e-04} & \textbf{8.50e-05} & \textbf{7.15e-01} & \textbf{1.05e-04} & \textbf{1.36e-06} & 5.25e+00&4.48e-01 \\
& PINNs& 8.04e-02 & 3.84e-01 & 9.68e-01 & 4.55e-03 & 3.40e-03 & \textbf{2.13e+00}&1.86e-01 \\
 
\bottomrule
\bottomrule
\end{tabular}
    \label{Müller-Brown loss}
    \end{table*}
    
\subsection{Alanine dipeptide}
As a final test we apply our iterative method to the canonical molecular benchmark of alanine dipeptide in vacuum, investigating the rare transition between the \(C_{7\mathrm{eq}}\) and \(C_{7\mathrm{ax}}\) conformers at \(T=300\) K~\cite{Bonati2020,khoo2019solving}. The full atomic system contains 22 atoms, so the phase-space dimension for the committor is nominally \(22\times 3 \times 2 = 132\). Direct approximation in this high-dimensional space is impractical, however, because the reactive dynamics are known to be dominated by a few collective degrees of freedom\cite{Bonati2020}. Guided by prior studies and by the arguments in Section~\ref{function_approximation}, we therefore perform a model reduction using internal dihedral coordinates as collective variables.
\begin{figure}[H]
        \centering
        \includegraphics[width=\linewidth]{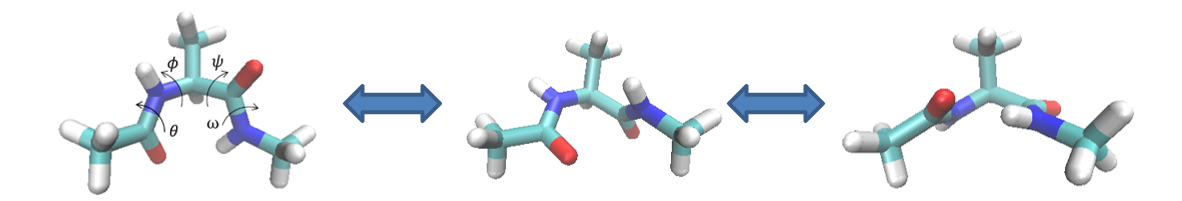}
        \caption{\label{ala2} Snapshots of alanine dipeptide conformations in metastable and transition states, shown from left to right as state A, the transition region, and state B. The relevant torsional angles of the peptide are depicted for state A.}
        
    \end{figure}

Specifically, we retain four dihedral angles \(\phi,\psi,\theta,\omega\) that quantify the relative orientations of the two peptide units with respect to the more rigid tetrahedron formed by the \(N\), \(C\), \(C_{\alpha}\), and \(C_{\beta}\) atoms (see Fig.~\ref{ala2}). Empirical evidence in the literature—most notably the effectiveness of \(\phi,\psi\) in enhanced-sampling schemes and transition-path studies—motivates this choice. Moreover, transition path sampling  studies have revealed that configurations in the reactive region exhibit an approximate linear anti-correlation between \(\phi\) and \(\theta\), i.e., \(\theta\approx -\phi\)\cite{Peter2000}. Building on the reduction strategy advocated in \cite{Kang2024}, we therefore adopt the ansatz that the committor predominantly depends on the reduced set of coordinates and the corresponding projected velocities:
\begin{equation}
    q(\mathbf{x},\mathbf{v}) \approx q(\phi,\theta,v_{\phi},v_{\theta})
\end{equation}
where \(v_{\phi}=\mathbf{v}\cdot\nabla_{\mathbf{x}}\phi\) and \(v_{\theta}=\mathbf{v}\cdot\nabla_{\mathbf{x}}\theta\). The justification for this manifold hypothesis and its implications for approximation error are discussed in Section~\ref{function_approximation}.

Figure~\ref{ala2_results} presents the learned committor for friction coefficients \(\gamma=5\) and \(25\ \mathrm{ps}^{-1}\). Panels (a) and (e) show the velocity-averaged committor \(\langle q\rangle_{v}\), which is in good qualitative agreement with prior results obtained under over-damped Langevin dynamics: the two metastable basins and the intervening reactive channel appear in the expected regions of the \(\phi\)–\(\psi\) plane. The near-linear anti-correlation between \(\phi\) and \(\theta\) is evident in these projections~\cite{Peter2000}, consistent with the commonly observed geometry of the transition tube. The remaining panels display committor slices at fixed values of the projected velocities. Unlike the Müller–Brown example, here we observe only modest variation of the committor with respect to velocity slices. We attribute this attenuation of velocity dependence primarily to the model reduction: by compressing the original 132-dimensional phase space onto the four-dimensional manifold \((\phi,\theta,v_{\phi},v_{\theta})\), much of the explicit influence of high-dimensional velocity components is averaged out or implicitly encoded in the reduced coordinates.

To validate the predictive quality of the learned committor, we perform a targeted shooting experiment. We sample 300 points satisfying \(|q-0.5|\leq 0.005\) and for each point launch 100 independent short trajectories to empirically estimate the committor probability. The distribution of these empirical committor estimates should be approximately Gaussian with mean 0.5 if the selected points lie in the true transition region and if sampling noise is predominantly statistical~\cite{lilinren2019}. Figure~\ref{ala2_results_1}(a,b,e,f) marks the selected transition-region points on the \(\phi\)–\(\psi\) and \(\phi\)–\(\theta\) projections; panels (c,d,g,h) report the histograms and Q–Q plots of the shooting results. The empirical distributions are close to Gaussian and centered near 0.5, demonstrating that our learned committor successfully identifies the transition region and produces calibrated probability estimates.

Taken together, these results indicate that the reduced-coordinate ansatz captures the essential reactive features of alanine dipeptide: the velocity-averaged committor agrees with established over-damped results, the expected near-linear anti-correlation between $\phi$ and $\theta$ are recovered, and direct shooting tests confirm the committor’s predictive accuracy. The modest velocity dependence observed after reduction highlights a practical trade-off: model reduction substantially simplifies the learning task and yields accurate transition-region identification, but it can also suppress explicit velocity effects that may be relevant in more weakly damped or strongly inertial regimes.

     \begin{figure}
    \includegraphics[width=\linewidth]{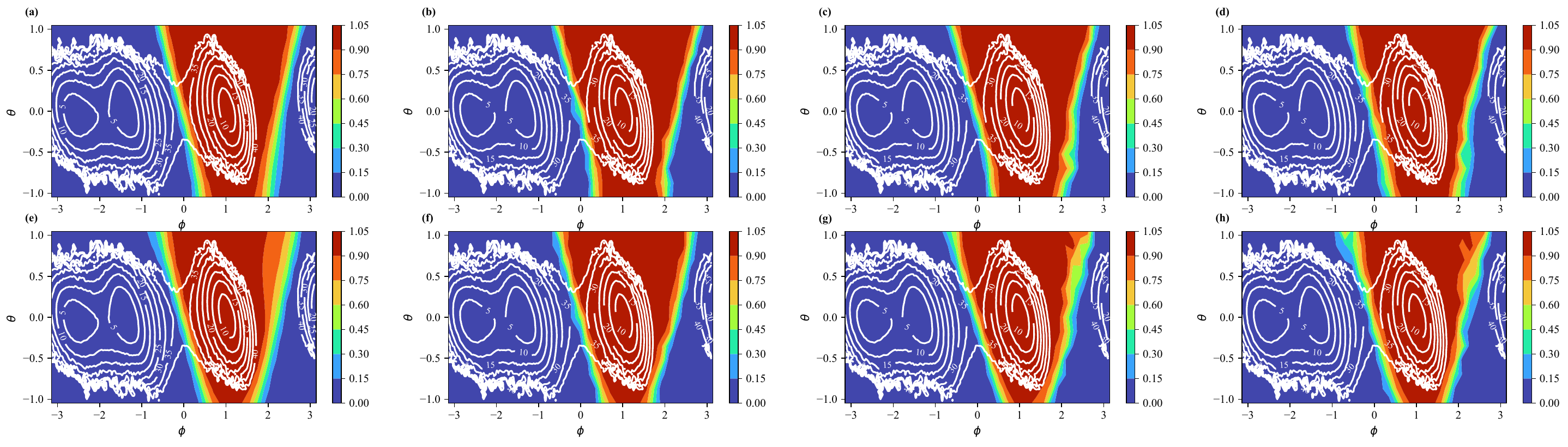}
    \caption{\label{ala2_results}Results for alanine dipeptide in vacuum. The white dashed line represents the free energy of the system and the contour plot shows the value of the committor function. Panels (a–d) correspond to $\gamma = 5$, while panels (e–h) correspond to $\gamma = 25$.  (a,e): the expectation of the learned committor function over velocity space $\langle q \rangle_v$; (b-d,e-h):committor function along $\phi$, $\theta$ at different fixed velocities $v_\phi$, $v_\theta$. 
    }
\end{figure}

     \begin{figure}
    \includegraphics[width=\linewidth]{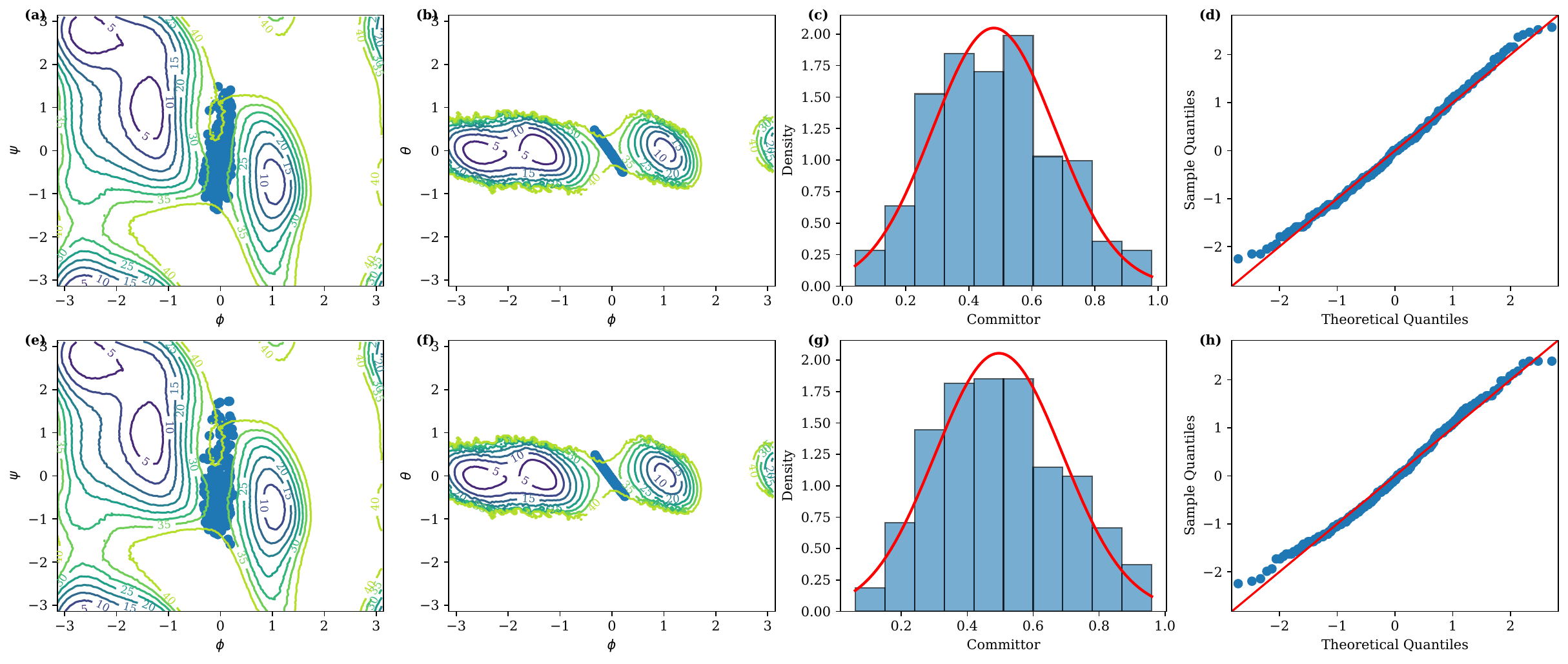}
    \caption{\label{ala2_results_1}
    Shooting results for alanine dipeptide in vacuum: panels (a–d) correspond to $\gamma = 5$, while panels (e–h) correspond to $\gamma = 25$.  (a,e): transition states in the $\phi$-$\psi$ plot; (b,f): transition states in the $\phi$-$\theta$ plot; (c,g): the histogram for the committor probability. For each point, 100 independent simulations are performed. The p-values of the Kolmogorov–Smirnov test for normal distribution are 0.6335 and 0.4038, correspondingly; (d,h): the Q-Q plots for the committor probability.}
\end{figure}

\section{Conclusion}
    In this work, we introduced a generalized CGW splitting framework to successfully compute the committor function for under-damped Langevin dynamics. We demonstrated that standard PINNs fundamentally struggle in complex energy landscapes, often falling victim to spectral ambiguities that cause severe topological failures near the transition interface. In contrast, our iterative neural variational solver bypasses these issues, reliably resolving the sharp $q \approx 0.5$ isosurface and delivering a highly accurate approximation of the true committor probability. A core contribution of this framework is the deployment of a targeted auxiliary potential, $\widetilde{U}$, which strategically shifts the sampling measure away from the metastable basins and directly into the reactive transition region. This enables precise identification of transition states and accurate estimation of reaction rates across various friction regimes. 
    
    Despite these computational advances, two primary challenges remain for future investigation. First, the current design of the auxiliary potential $\widetilde{U}$ relies on the heuristic principle that sampling must be aggressively focused on the transition zone. Developing a fully adaptive, mathematically rigorous strategy for constructing this auxiliary measure will further maximize the algorithm's efficiency. Second, the spatial operator $\KKK$ effectively serves as a preconditioner for the system's infinitesimal generator, $\LLL$. While we have empirically demonstrated its ability to stabilize the optimization, establishing an automated, theoretically grounded protocol for selecting the optimal hyperparameters ($\omega_1$ and $\omega_2$) to strictly maximize the convergence rate remains an open problem in operator analysis. We leave these theoretical refinements for future study.

\section*{ACKNOWLEDGMENTS}
TL and ZL acknowledge the support from National Key R\&D Program of China under grant 2021YFA1003301, and National Science Foundation of China under grant 12288101. XZ acknowledges the support of General Research Funds from the Research Grants Council of the Hong Kong Special SAR, China (Project No. 11304525).

\section*{CRediT authorship contribution statement}
Tiejun Li: Conceptualization, Methodology, Writing – review and editing.

Zhaoqi Liu: Methodology, Validation, Formal analysis, Writing – original draft.

Xiang Zhou: Conceptualization, Methodology, Writing – review and editing.

\appendix

\section{Detailed settings and additional results for numerical experiments\label{details}}
\subsection{Double well potential}
We evaluate our method on this potential by setting $\kbt = 0.1$ and examining the friction coefficients $\gamma \in \{0.2, 1, 5, 25\}$. For the $k$-th iterative step, we employ a neural network $q^{k}_\theta(\bm{x}, \bm{v})$ featuring an architecture of [6, 20, 20, 1] nodes. Optimization is performed using the Adam optimizer with a learning rate of $l^{k}$. We train each subproblem using a dataset of size $10^6$, sampled uniformly in the range $[-1, 1]$ and assigned a weight $w_k \propto \beta^k \Vert \nabla_{\bm{x}} q^k_{\theta} \Vert^{\alpha^k}$. This weighting scheme concentrates the optimization process directly on the transition region. We define the hyperparameters for the spatial operator $\KKK$ as $\omega^k_1$ and $\omega^k_2$, and let $N_k$ denote the number of training epochs for the $k$-th subproblem. These hyperparameter settings are summarized in Table~\ref{1d_hyper}. For comparison, we train a baseline PINN using the identical network architecture for 1000 epochs with a learning rate of $10^{-3}$, followed by 8000 epochs at $10^{-4}$. We employ an $L^2$ residual loss function defined as $\frac{1}{N}\sum_{i=1}^{N} \Vert \LLL q_{\theta}\Vert^2$. The training points are generated by sampling $x_1$ uniformly from the interval $[-1,1]$ and drawing $v_1$ from the normal distribution $\mathcal{N}(0,\kbt)$.The reference solution is computed via a finite difference method (FDM) with step sizes $\Delta x_1 = 0.05$ and $\Delta v = 0.01\sqrt{\kbt}$.
\begin{table}[H]
    \centering
    \caption{Hyperparameters for the Double well potential}
    \label{1d_hyper}
    \begin{tabular}{c|c|c|c|c|c|c}
    \toprule
        Num of steps &$l^{k}$& $N_k$ & $\omega^k_1$ & $\omega_2^k$ & $\alpha^k$ & $\beta^k$  \\
    \midrule
        50& $10^{-3}$ & 20 & 0.2 & 0.01 & 0.9 & 0.3  \\
        200 & $10^{-4}$ & 40 & 0.2 & 0.01 & 0.9 & 0.3  \\
    \bottomrule
    \end{tabular}
\end{table}

Figures~\ref{1d_committor_slice} and \ref{1d_committor_pinn_slice} present the numerical results for both our method and the PINN across the evaluated $\gamma$ regimes. For each case, we plot five planar cross-sections as functions of $x_1$ and $v_1$ by holding the remaining variables fixed. As $\gamma$ increases, the $q \approx 0.5$ isosurface becomes noticeably less tilted, indicating that velocity exerts less influence on the committor function—an observation that is highly consistent with predictions from asymptotic analysis. Furthermore, when compared to the PINN baseline, our method yields a significantly thinner, sharper transition interface that closely matches the FDM reference solution. This confirms that our approach effectively resolves the critical dynamical information within the transition region.
\begin{figure}[H]
\centering
    \includegraphics[width=0.6\linewidth]{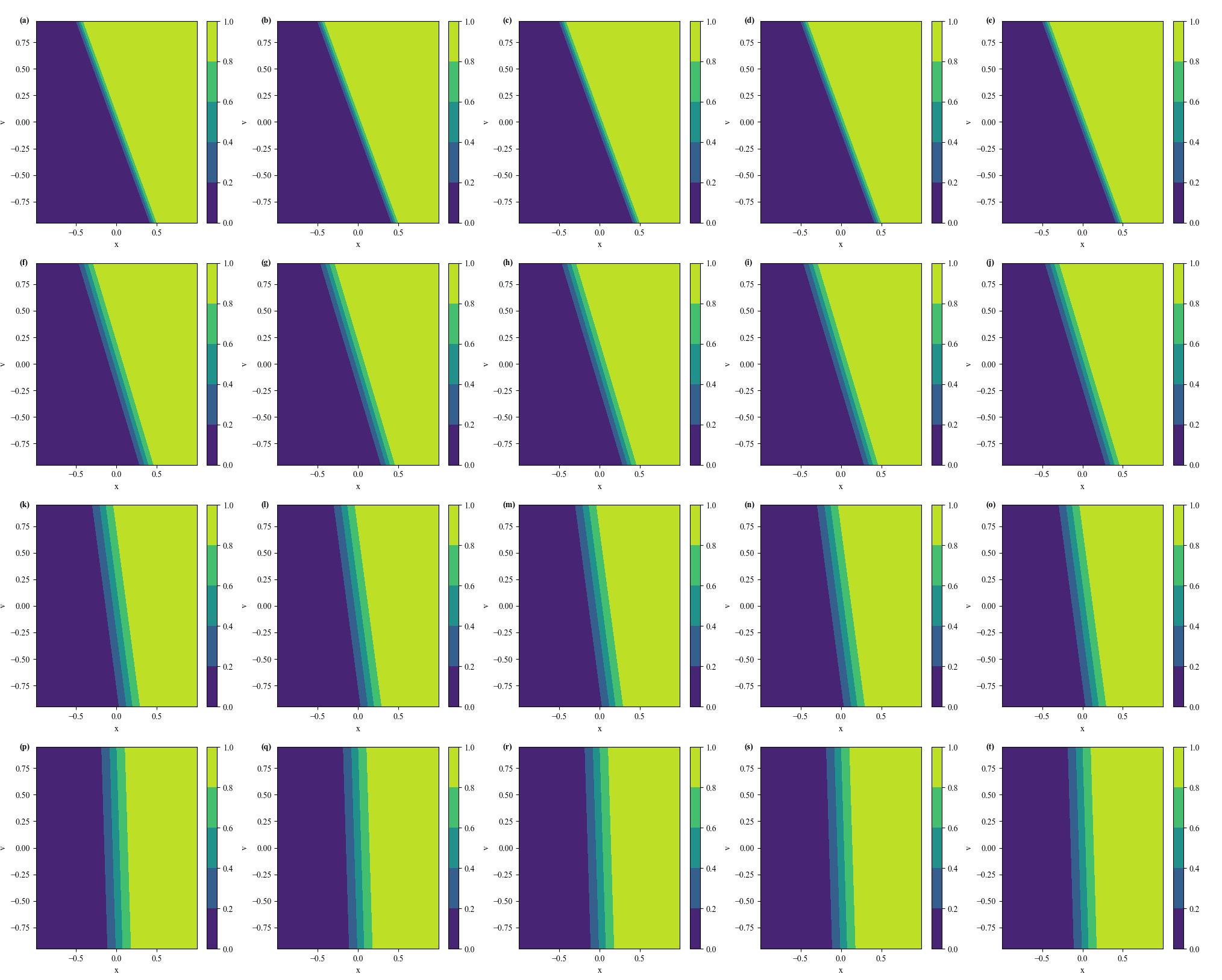}
    \caption{\label{1d_committor_slice}Committor function learned from our method. Each plane is a function of $x_1,v_1$ with fixed remaining variables. (a-e): $\gamma = 0.2$; (f-j): $\gamma = 1$; (k-n): $\gamma = 5$; (p-t): $\gamma = 25$}
\end{figure}

\begin{figure}[H]
\centering
    \includegraphics[width=0.6\linewidth]{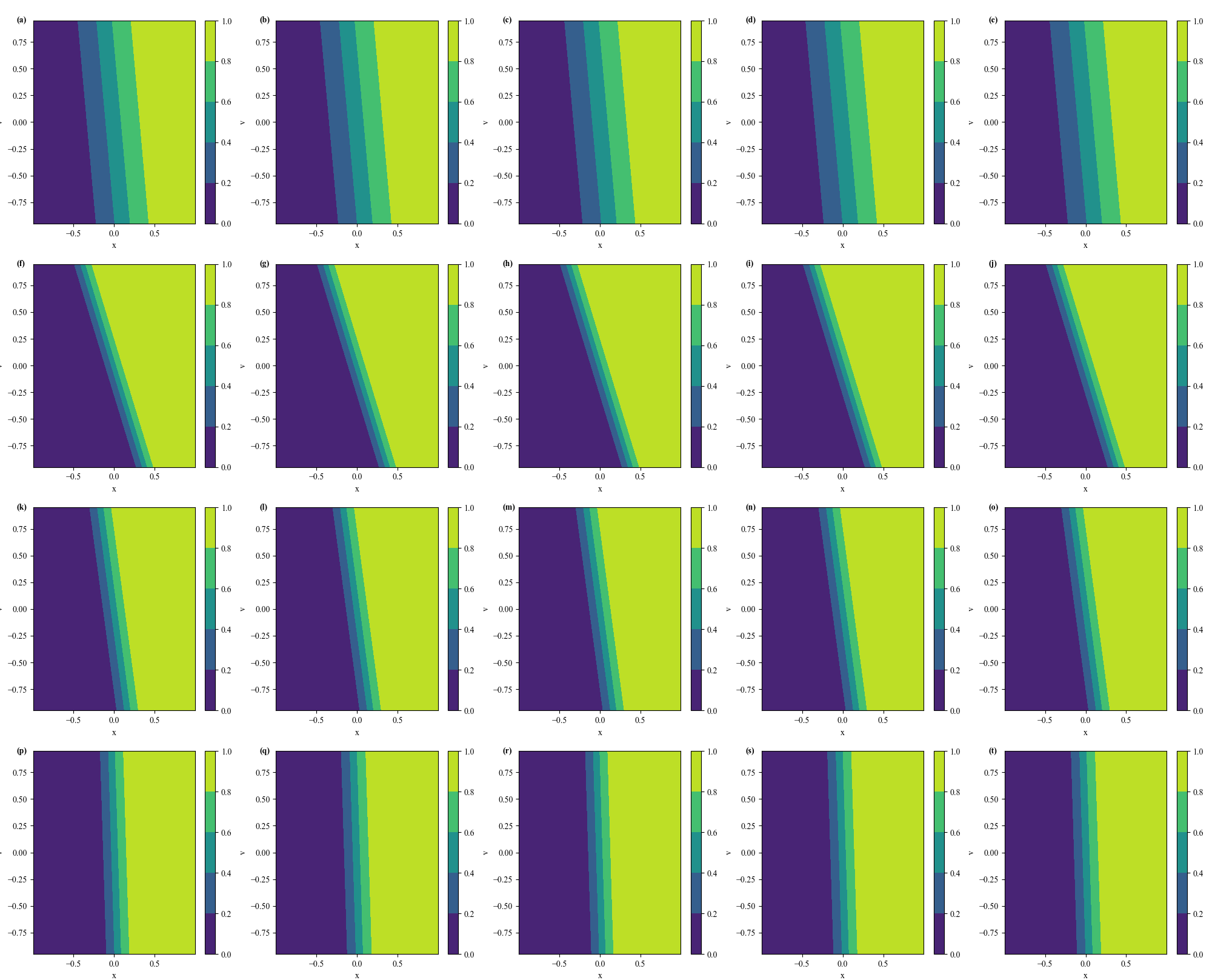}
    \caption{\label{1d_committor_pinn_slice}Committor function learned from PINNs. Each plane is a function of $x_1,v_1$ with fixed remaining variables. (a-e): $\gamma = 0.2$; (f-j): $\gamma = 1$; (k-n): $\gamma = 5$; (p-t): $\gamma = 25$}
\end{figure}

\subsection{Müller-Brown potential}

We evaluate our method on this potential by setting $\kbt = 5$ and examining friction coefficients $\gamma \in \{1, 5, 25\}$. For the $k$-th iterative step, we parameterize the committor function using a neural network, $q^{k}_\theta(\bm x, \bm v)$, featuring a $[4, 8, 64, 64, 64, 64, 8, 1]$ node architecture. Optimization is performed using the Adam optimizer with a step-dependent learning rate, $l^{k}$. Each subproblem is trained on a dataset of size $N_d$. The data is initially sampled uniformly over the spatial domain $[-1.5, 1.2] \times [-0.2, 2]$. To focus the learning capacity, each sample is assigned a composite weight $w = \beta^{k}_{\text{top-k}} w^{k}_{\text{top-k}} + \beta^{k}_{\text{grad}}w_{\text{grad}} + (1 - \beta^{k}_{\text{top-k}} - \beta^{k}_{\text{grad}}w_{\text{grad}}) \cdot \frac{1}{N_d}$. Here, $w^{k}_{\text{top-k}}$ is determined via the top-$k$ algorithm (an adaptive sampling strategy designed for PINNs), and $w_{\text{grad}} \propto \Vert \nabla_{\xx} q \Vert^2$. Furthermore, because the training explicitly evaluates the potential gradient $\nabla U$, we exclude points with potential energy greater than zero from numerical consideration to maintain stability. The hyperparameters defining the spatial operator $\KKK$ are denoted as $\omega^k_1$ and $\omega^k_2$, and the number of training epochs for the $k$-th subproblem is denoted as $N_k$. These configurations are summarized in Table~\ref{muller_hyper}. For comparison, we train a baseline PINN model for 100,000 epochs with a learning rate of $10^{-4}$, followed by 100,000 epochs at $5 \times 10^{-5}$, and finally 200,000 epochs at $10^{-5}$. We employ an $L^2$ residual loss function defined as $\frac{1}{N}\sum_{i=1}^{N} \Vert \LLL q_{\theta}\Vert^2$. The training points are generated by sampling $(x_1,x_2)$  uniformly in the region $\{(x_1,x_2):U(x_1,x_2)<100)\}$ and drawing $(v_1,v_2)$ from the normal distribution $\mathcal{N}(0,\kbt)$.The reference ground-truth solution is generated via 100 independent brute-force simulations using the BAOAB algorithm~\cite{leimkuhler2013rational,BAOAB2013} with a time step of $2 \times 10^{-5}$.
\begin{figure}[H]
    \includegraphics[width = \linewidth]{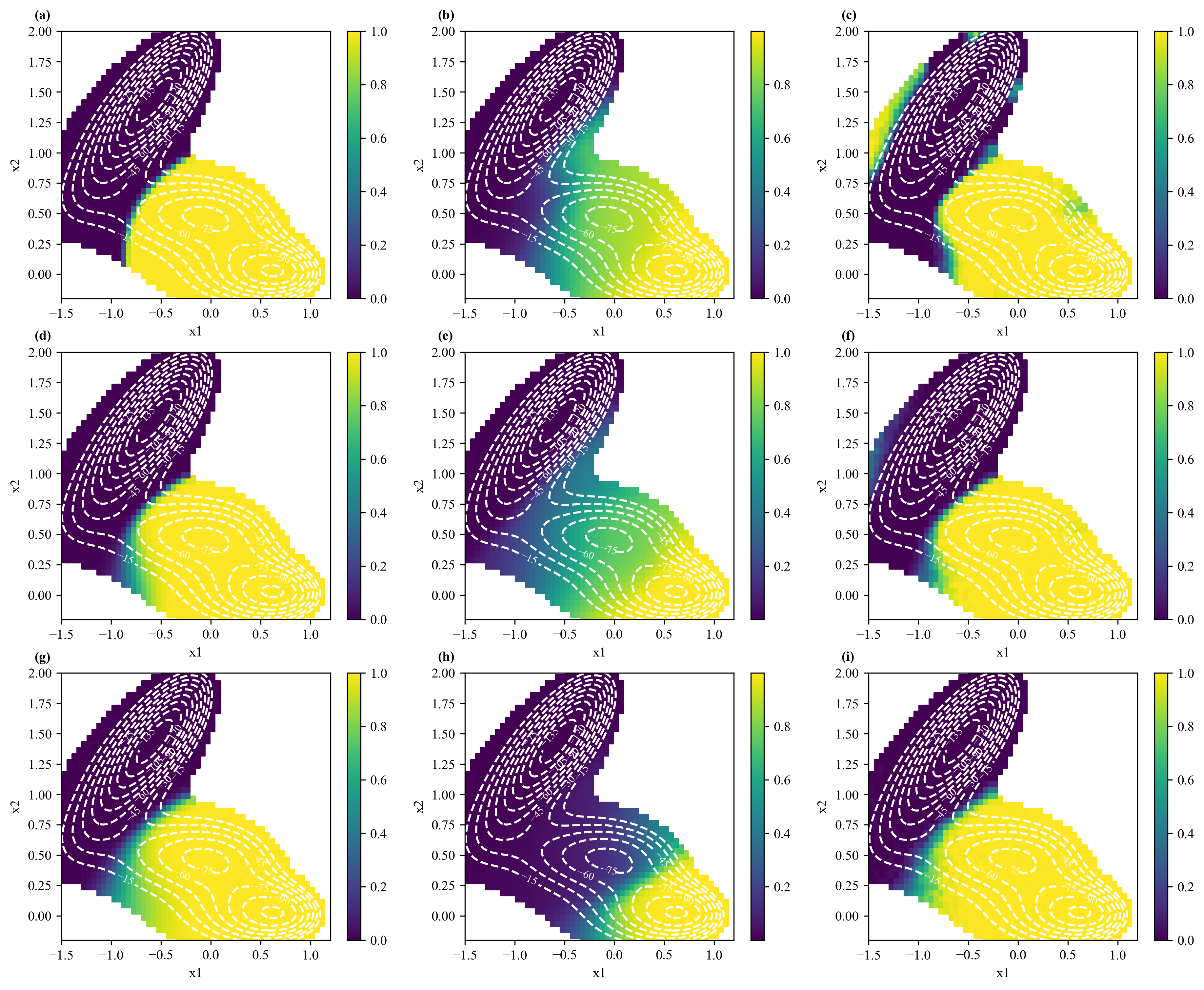}
    \caption{\label{muller_committor_gammas} Comparison of the results for the Müller-Brown potential from different methods: plane (a,d,g): Ours; plane (b,e,h): PINNs; plane (c,f,i) brute force simulation. Different rows correspond to different $\gamma$ values: plane (a,b,c): $\gamma = 1$; plane (d,e,f): $\gamma=5$; plane (g,h,i): $\gamma = 25$. The white dashed line represents the potential energy. }
\end{figure}

In Fig.\ref{muller_committor_gammas}, we present the numerical solution from our method and PINNs at $\kbt = 5$ and $\gamma = 1,5,25$, associated with the reference solution from brute force simulation (committor analysis). For a specific velocity, as $\gamma$ grows larger, the transition interface $q\approx 0.5$ becomes wider. This agrees with physical intuition: a larger friction coefficient speeds up the convergence of the velocity to its stationary distribution and reduces the influence of the initial velocity on the committor function.

\begin{table}[H]
    \centering
    \caption{Hyperparameters for the Müller-Brown potential}
    \label{muller_hyper}
    \begin{tabular}{c|c|c|c|c|c|c}
    \toprule
        Num of steps &$lr^{k}$& $N_k$ & $\omega^k_1$ & $\omega_2^k$ & $\beta_{grad}^k$ & $\beta_{\text{top-k}}$  \\
    \midrule
        200& $1\times10^{-3}$ & 20 & 0.1 & 0.1 & 0.9 & 0.05  \\
        200 & $1\times 10^{-4}$ & 20 & 0.1 & 0.1 & 0.9 & 0.05  \\
        1000 & $5\times 10^{-5}$ & 60 & 2 & 2 & 0.9 & 0.05  \\
        1000 & $1\times 10^{-5}$ & 60 & 2 & 2 & 0.9 & 0.05  \\
    \bottomrule
    \end{tabular}
\end{table}

\subsection{Alanine dipeptide}
Molecular dynamics (MD) simulations of alanine dipeptide (Ace-Ala-Nme) in vacuo were executed using GROMACS-2025.3~\cite{Hess2008,10.1093/bioinformatics/btt055}, augmented with the PLUMED plugin~\cite{TRIBELLO2014604}. The system was parameterized using the Amber99-SB force field~\cite{amber99-SB} and integrated with a 2 fs time step. Phase-space exploration was driven by Langevin dynamics, maintaining a target temperature of 300 K via a damping coefficient $\gamma_i = \frac{m_i}{\tau_t}$, parameterized with a coupling time constant of $\tau_t = 0.05$ ps. To learn the full phase-space committor function $q(\bm{x}, \bm{v})$ governing the under-damped Langevin dynamics, we adopt a two-stage approach~\cite{Kang2024}. First, we pre-train a baseline committor function $q_0(\bm{x})$ derived from the corresponding over-damped Langevin dynamics. Subsequently, this $q_0$ informs a 12 ns biased MD simulation employing the potential $V = -\lambda \log(\Vert \nabla_{\xx} q_0 \Vert^2 + \epsilon)$. The scaling hyperparameter was set to $\lambda = 3.5$, which provided the most efficient sampling. The harvested trajectory data is then filtered to exclude configurations mapped to the defined metastable states $A$ and $B$, ensuring the final dataset strictly comprises transition-pathway configurations.
Molecular dynamics (MD) simulations of alanine dipeptide (Ace-Ala-Nme) in vacuo were executed using GROMACS-2025.3~\cite{Hess2008,10.1093/bioinformatics/btt055}, augmented with the PLUMED plugin~\cite{TRIBELLO2014604}. The system was parameterized using the Amber99-SB force field~\cite{amber99-SB} and integrated with a 2 fs time step. Phase-space exploration was driven by Langevin dynamics, maintaining a target temperature of 300 K via a damping coefficient $\gamma_i = \frac{m_i}{\tau_t}$, parameterized with a coupling time constant of $\tau_t = 0.05$ ps. To learn the full phase-space committor function $q(\bm{x}, \bm{v})$ governing the under-damped Langevin dynamics, we adopt a two-stage approach~\cite{Kang2024}. First, we pre-train a baseline committor function $q_0(\bm{x})$ derived from the corresponding over-damped Langevin dynamics. Subsequently, this $q_0$ informs a 12 ns biased MD simulation employing the potential $V = -\lambda \log(\Vert \nabla_{\xx} q_0 \Vert^2 + \epsilon)$. The scaling hyperparameter was set to $\lambda = 3.5$, which provided the most efficient sampling. The harvested trajectory data is then filtered to exclude configurations mapped to the defined metastable states $A$ and $B$, ensuring the final dataset strictly comprises transition-pathway configurations.

During the $k$-th iterative optimization step, the under-damped committor $q^{k}_\theta(\bm{x}, \bm{v})$ is represented by a deep neural network featuring a $[4, 8, 256, 256, 1]$ architecture. The network weights are updated using the Adam optimizer, governed by a learning rate $l^{k}$. The tuning parameters for the spatial operator $\KKK$, denoted as $\omega^k_1$ and $\omega^k_2$, alongside the corresponding epoch counts $N_k$ for each subproblem, are comprehensively tabulated in Table~\ref{ala2_hyper}.
\begin{table}[H]
    \centering
    \caption{Hyperparameters for Alanine dipeptide}
    \label{ala2_hyper}
    \begin{tabular}{c|c|c|c|c}
    \toprule
        Num of steps &$lr^{k}$& $N_k$ & $\omega^k_1$ & $\omega_2^k$   \\
    \midrule
        400& $1\times10^{-3}$ & 40 & 4 & 4   \\
        400 & $1\times 10^{-4}$ & 80 & 4 & 4 \\
        400 & $5\times 10^{-5}$ & 80 & 4 & 4   \\
        400 & $1\times 10^{-4}$ & 80 & 10 & 10  \\
        2000 & $5\times 10^{-5}$ & 80 & 10 & 10  \\
    \bottomrule
    \end{tabular}
\end{table}

The choice of sampling strategy significantly impacts the quality of the learned committor function. Fig.~\ref{ala2_gamma25_biased} and \ref{ala2_gamma25_metad} compare models trained on two different datasets—one generated via our proposed biased simulation and the other via metadynamics using $\phi$ and $\psi$ as collective variables (CVs)—while keeping all hyperparameters fixed. The results demonstrate that our biased simulation dataset enables the network to learn a sharply defined, smooth transition interface. In contrast, training on the metadynamics dataset yields noticeably inferior performance, failing to adequately resolve the sharp boundary between the metastable states.

\begin{figure}
    \includegraphics[width=\linewidth]{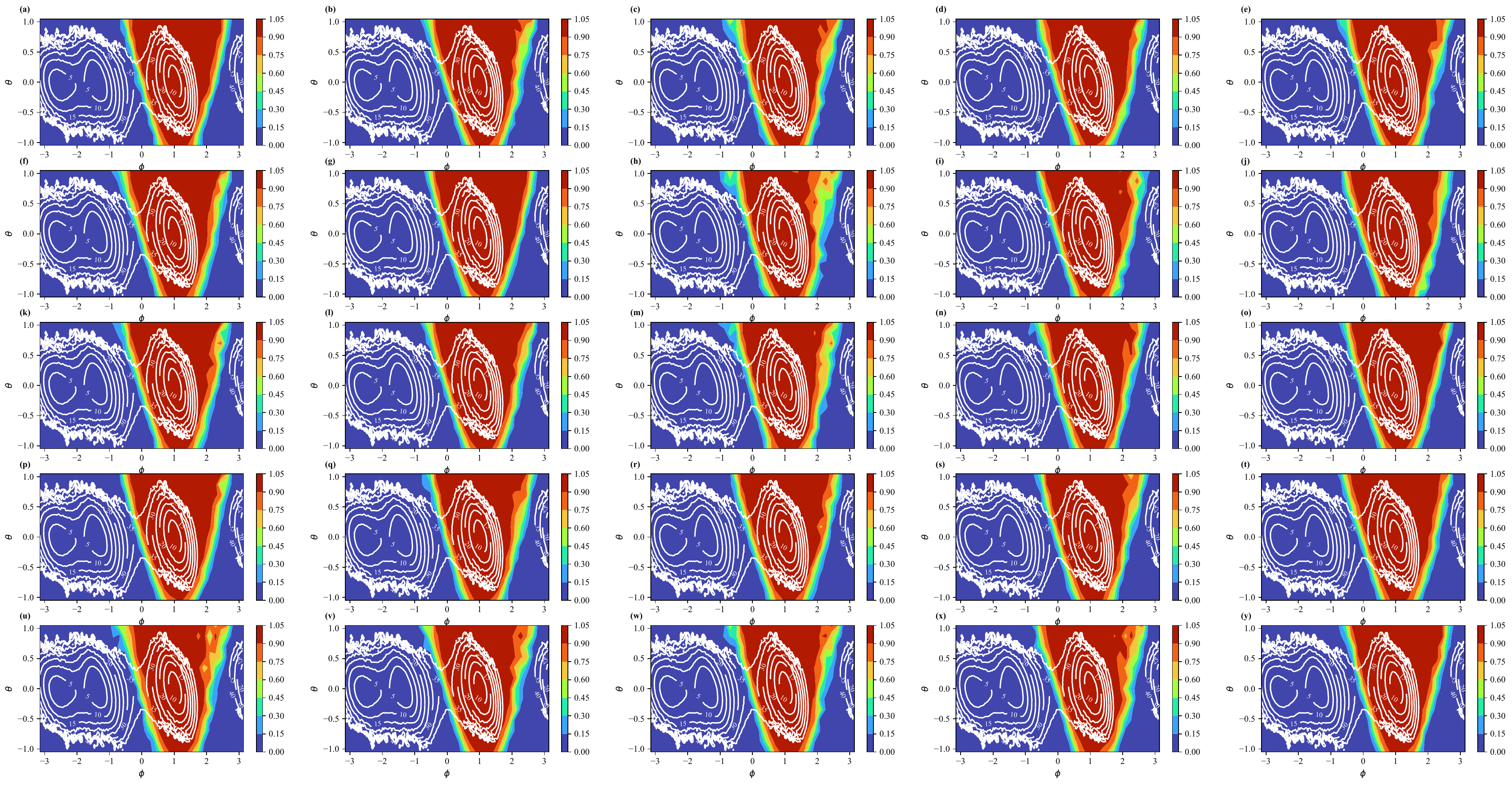}
    \caption{\label{ala2_gamma25_biased}Committor function learned using data from biased sampling with biased potential $V = -\lambda\log (\Vert q_0\Vert^2+\epsilon)$ at $\gamma = 25$. The contour plot shows the committor function projected onto the $\phi$-$\theta$ space. Each plane shows the committor function at a fixed velocity.}
\end{figure}

\begin{figure}
    \includegraphics[width=\linewidth]{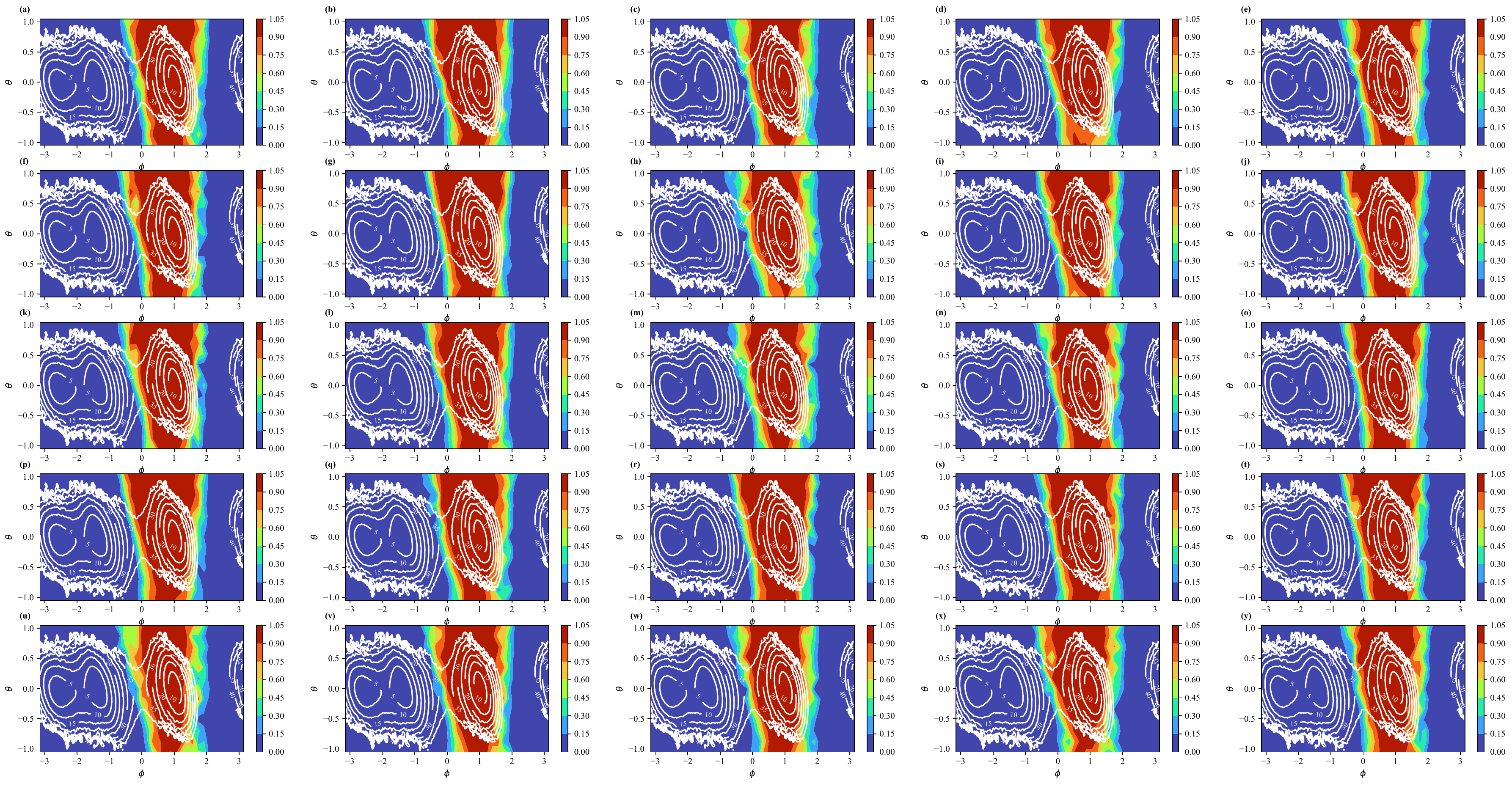}
    \caption{\label{ala2_gamma25_metad}Committor function learned using data from metadynamics at $\gamma = 25$. The contour plot shows the committor function projected onto the $\phi$-$\theta$ space. Each plane shows the committor function at a fixed velocity.}
\end{figure}

\bibliography{bibliography}
\bibliographystyle{model1-num-names}

\end{document}